\documentclass{article}

\usepackage{arxiv}

\usepackage{graphics} 
\usepackage{times} 
\usepackage{amsmath} 
\usepackage{amsfonts}
\usepackage{amsthm}
\newtheorem{thm}{Theorem}
\newtheorem{lemma}{Lemma}
\newtheorem{corollary}{Corollary}
\newtheorem{remark}{Remark}

\newtheorem{prop}{Proposition}
\renewenvironment{proof}{{\bf Proof.}}{\qed}
\usepackage{mathtools}
\usepackage{bm}
\usepackage{booktabs}
\usepackage{enumerate}
\usepackage{siunitx}
\newcommand{\Rnum}[1]{\uppercase\expandafter{\romannumeral #1\relax}}
\newcommand{\rnum}[1]{\lowercase\expandafter{\romannumeral #1\relax}}
\usepackage{accents}

\usepackage{caption}
\usepackage[font=normalsize]{subcaption}
\usepackage{microtype}
\usepackage[numbers]{natbib}

\newcommand{\ip}[2]{\langle{#1},{#2}\rangle}

\title{\protect\textnormal{Robust Control Design and Analysis Based on Lifting Linearization of Nonlinear Systems Under Uncertain Initial Conditions}}

\author{
Sourav Sinha\thanks{Corresponding author} \\
Kevin T. Crofton Department of Aerospace and Ocean Engineering \\
Virginia Tech \\
Blacksburg, VA 24061, USA \\
\texttt{srvsinha@vt.edu} \\
\And
Mazen Farhood \\
Kevin T. Crofton Department of Aerospace and Ocean Engineering \\
Virginia Tech \\
Blacksburg, VA 24061, USA \\
\texttt{farhood@vt.edu}
}
\date{}
\renewcommand{\headeright}{}
\renewcommand{\undertitle}{}
\renewcommand{\shorttitle}{}

\begin{document}
\maketitle

\begin{abstract}
This paper presents a robust control synthesis and analysis framework for nonlinear systems with uncertain initial conditions. First, a deep learning-based lifting approach is proposed to approximate nonlinear dynamical systems with linear parameter-varying (LPV) state-space models in higher-dimensional spaces while simultaneously characterizing the uncertain initial states within the lifted state space. Then, convex synthesis conditions are provided to generate full-state feedback non-stationary LPV (NSLPV) controllers for the lifted LPV system. A performance measure similar to the $\ell_2$-induced norm is used to provide robust performance guarantees in the presence of exogenous disturbances and uncertain initial conditions. The paper also includes results for synthesizing full-state feedback linear time-invariant controllers and output feedback NSLPV controllers. Additionally, a robustness analysis approach based on integral quadratic constraint (IQC) theory is developed to analyze and tune the synthesized controllers while accounting for noise associated with state measurements. This analysis approach characterizes model parameters and disturbance inputs using IQCs to reduce conservatism. Finally, the effectiveness of the proposed framework is demonstrated through two illustrative examples.
\end{abstract}

\keywords{Lifted models\and LPV control\and robust control\and integral quadratic constraints\and uncertain initial conditions}

\section{Introduction} \label{sec_intro}

Control design and analysis methods for linear systems are typically far less computationally intensive compared to their counterparts for nonlinear systems \cite{TSUKAMOTO2021135,LuDoyle97,megretski1997system,veenman2016robust}. By far the most common approach to simplifying a nonlinear system for control design purposes is to linearize the dynamic equations about a trajectory in order to generate a linear  approximation of the nonlinear system. For time-invariant nonlinear systems controlled about an equilibrium point, this approximation will be in the form of a linear time-invariant (LTI) model. However, LTI models derived using traditional Jacobian linearization techniques are valid only locally, within a small envelope around the equilibrium point. Consequently, the stability and performance guarantees generated for these linear models may not necessarily apply to the actual nonlinear system, particularly when the operating envelope is large, possibly due to significant initial state perturbations or the presence of exogenous disturbances. To overcome this limitation, we approximate nonlinear dynamical systems using linear parameter-varying (LPV) models within an envelope around a specific operating point. These LPV models have state-space matrix-valued functions that explicitly depend on the so-called scheduling parameters, which, in this work, are defined as nonlinear functions of the system state. This enables LPV systems to represent nonlinear dynamics more accurately over a larger operating envelope compared to LTI systems. 

Commonly used approaches for LPV modeling of nonlinear systems include linearization around a family of equilibria \cite{muniraj2017path, Fry2019} and quasi-LPV methods \cite{Toth2010, marcos2004development}. The former is suitable when the system has multiple operating points that can be parameterized as rational functions of time-varying parameters. This method yields an LPV system that approximates the nonlinear dynamics around the operating points; however, its validity remains confined to small envelopes surrounding these points. In contrast, the quasi-LPV approach offers an exact LPV representation but is restricted to a limited class of nonlinear systems \cite{Toth2010}. Moreover, for complex systems, quasi-LPV methods may yield systems with a large number of scheduling parameters that depend on the system's state and/or input, introducing conservatism into the synthesis and analysis problems, as shown in our previous work \cite{sinha2021lft}.

In recent years, Koopman operator theory (KOT) has attracted significant attention for its applications in data-driven modeling of nonlinear dynamical systems \cite{lusch2018deep, williams2015data, li2017extended, otto2019linearly}. The Koopman operator, introduced by B. O. Koopman \cite{koopman1931hamiltonian}, provides a linear embedding of unforced, time-invariant nonlinear systems within a higher-dimensional space of functions, called observables. Unlike linearized models that are only valid locally, the Koopman operator provides a globally accurate linear representation of nonlinear system dynamics across the entire state space. However, the Koopman operator is typically infinite-dimensional. Finite-dimensional linear embeddings have been identified for only a limited number of nonlinear systems \cite{brunton2016koopman}, and for systems with multiple fixed points, the finite dimensional linear representations, if existent, are shown to be discontinuous \cite{bakker2019learning, liu2023non}.

Consequently, data-driven lifting methods are commonly used to learn finite-dimensional linear approximations of nonlinear systems within the basin of attraction of a fixed point. These methods include extended dynamic mode decomposition (EDMD) \cite{williams2015data} and deep learning-based techniques \cite{lusch2018deep, otto2019linearly}, which are tied to KOT. While EDMD methods are computationally efficient compared to deep-learning approaches, they require a priori selection of the set of observables, presenting a challenge in choosing an optimal set that minimizes the error induced due to finite subspace approximation. In contrast, deep learning methods use a trainable neural network to parameterize the observables, which are learned simultaneously with the system matrices in the lifted state space. This enhances approximation accuracy even with a reduced number of observables \cite{li2017extended}, albeit at the cost of increased computational complexity. The key idea involves learning a nonlinear mapping that ``lifts" the original system state to a higher-dimensional space, where the dynamics of the ``lifted” system become approximately linear. This approach has also been extended to actuated nonlinear systems. In \cite{surana2016koopman,goswami2017global}, it is demonstrated that control-affine nonlinear systems can be exactly represented by infinite-dimensional bilinear systems. In \cite{iacob2024koopman}, it is shown that the lifted bilinear representation of actuated nonlinear systems can be interpreted as LPV models, and in \cite{eyuboglu2024data}, the EDMD approach is utilized to learn lifted LPV approximations of nonlinear systems. In this work, we adopt a deep learning framework, parameterizing both the ``lifting" function and the scheduling parameter map with a neural network. This enables the generation of effective lifted LPV approximations with a relatively small number of parameters. While our work does not directly deal with the Koopman operator, for completeness a brief overview of KOT will be provided, along with results on lifted models from the literature that are relevant to the work done in this paper.

The control design problem for nonlinear systems using bilinear lifted models has been extensively studied in the literature. For instance, model predictive control strategies based on bilinear lifted models have been proposed in \cite{folkestad2021koopman, narasingam2023data}. Control Lyapunov function-based approaches for designing stabilizing state-feedback controllers have also been provided in previous works \cite{huang2018feedback, zinage2023neural}. Additionally, stabilizing state-feedback controllers with a linear dependence on the lifted state have been synthesized by solving linear matrix inequalities (LMIs) \cite{sinha2022data, strasser2023robust}. LMI-based approaches for synthesis of nonlinear state-feedback controllers have also been proposed in the literature \cite{strasser2024koopman, strasser2024safedmd}. In \cite{eyuboglu2024data}, robust control design for lifted LPV models is addressed; however, its scope is restricted to designing LTI controllers without performance guarantees and under the assumption that the LPV dynamics are open-loop bounded-input, bounded-output (BIBO) stable. The LPV model is particularly appealing because gain-scheduling or LPV control strategies can be developed for nonlinear systems based on LPV plant models \cite{rugh2000research, mohammadpour2012control}. An LPV system, with at most a rational dependence on the parameters, can also be expressed as a linear fractional transformation (LFT) on uncertainties, i.e., an interconnection of a nominal LTI system and a perturbation operator consisting of uncertainties. Robust synthesis techniques \cite{packard1994gain, apkarian1995convex, farhood2008controlNSLPV} and robustness analysis methods \cite{lu1996stabilization, megretski1997system} can then be utilized to generate LPV controllers for the plant and compute the robust performance level (upper bound on the worst-case performance) of the resulting closed-loop system, respectively.

In this work, we address the robust control design problem for nonlinear systems with uncertain initial conditions, where the initial state value is constrained to lie within a predefined ellipsoid. Existing lifting-based control approaches \cite{narasingam2023data, zinage2023neural, sinha2022data, strasser2024safedmd} provide stability guarantees for bilinear lifted models under non-zero initial conditions. However, in these works, the region of attraction is defined in the higher-dimensional state space and is a byproduct of the synthesis approach, rather than being defined a priori. The objective of this work is to design controllers that provide both stability and performance guarantees for the lifted system in the presence of exogenous disturbances. Additionally, these guarantees should hold for all possible initial state values within an ellipsoid in the lifted state space, which is derived from the predefined ellipsoid in the original state space. As in aforementioned works, the lifted system is only an approximation of the nonlinear system, and hence, the stability and performance guarantees obtained from analysis for the lifted system cannot be extended to the nonlinear system unless the approximation errors are appropriately incorporated into the analysis.

Deriving error bounds for data-driven, finite-dimensional approximations of the Koopman operator is an active research area but not the focus of this work. For completeness, we provide a brief overview of key contributions from the literature. Under some assumptions, N{\"u}ske et al. \cite{nuske2023finite} established probabilistic bounds on the approximation error arising from the finite number of data samples used in EDMD. Schaller et al. \cite{schaller2023towards} extended this work to include the projection error due to a finite number of observables. Similar work was carried out for kernel EDMD \cite{philipp2023error}. In addition, error bounds proportional to the state and control, rather than uniform bounds, have been recently derived \cite{strasser2023robust,strasser2024koopman,bold2024data}. These works leverage approximate models with error bounds to design controllers for nonlinear systems with stability guarantees. However, they only account for the error due to the finitely many data points used in EDMD and assume the existence of an exact finite-dimensional linear approximation of the nonlinear system. Other works on designing stabilizing Koopman-based controllers for nonlinear systems \cite{strasser2023robust, mamakoukas2022robust} account for modeling errors from both finite data and observables by introducing an additive signal in the lifted state equation. In contrast, Eyuboglu et al. \cite{eyuboglu2024data,eyuboglu2024koopman} add a non-parametric dynamic uncertainty to the linear approximate model under the assumption that it is bounded.

We propose a deep learning approach for learning lifted LPV approximations of the nonlinear system while simultaneously characterizing the uncertain initial conditions within the higher-dimensional state space. The scheduling parameters are defined as nonlinear functions of the lifted state. We build upon the method proposed in our previous work \cite{farhood2021lpv} to synthesize static state-feedback nonstationary LPV (NSLPV) controllers for the lifted LPV model. NSLPV systems have state-space matrix-valued functions with explicit dependence on both the scheduling parameters and time. They constitute a fairly general class of systems, which include LTI, linear time-varying (LTV), and standard (stationary) LPV systems \cite{farhood2008controlNSLPV, farhood2007model}. When interpreted in the original state space, these NSLPV controllers become nonlinear with explicit dependence on time.

While a stationary LPV controller can be synthesized for the LPV model using the proposed approach, an NSLPV controller has the potential to enhance closed-loop performance of systems with non-zero initial conditions, as demonstrated in our previous works \cite{farhood2021lpv, farhood2008control}. We solve the synthesis conditions under the assumption that the system states are exactly measurable. This assumption facilitates the construction of static controllers instead of dynamic ones, which can be computationally intensive to design and implement, particularly when the dimension of the lifted state is very large. Additionally, we propose a robustness analysis approach based on integral quadratic constraint (IQC) theory \cite{megretski1997system} to compute the robust performance level of the resulting closed-loop NSLPV systems. IQC analysis is a flexible and efficient tool for analyzing uncertain systems expressed in LFT form. One of its key strengths lies in its ability to handle a wide range of uncertainties, including static and dynamic, LTI and LTV perturbations, nonlinear (sector-bounded, norm-bounded, slope-restricted, and passive) uncertainties, and delays \cite{megretski1997system, veenman2016robust}. Moreover, it allows the use of signal IQC multipliers to characterize disturbance sets effectively. Signal IQC multipliers for disturbances such as white noise, band-limited signals, constant signals, monotonic signals (increasing or decreasing), and signals restricted to a finite time interval are given in \cite{jonsson1999iqc, fry2021robustness}. 
The proposed analysis approach extends the results of our previous work \cite{farhood2024robustness} to handle signal IQCs and nonlinear uncertainties. We use this approach to analyze the synthesized static controllers while accounting for the measurement noise. A tuning routine based on IQC analysis is also provided to guide the control design process. Finally, we apply the proposed approach to design controllers for two complex examples, namely, a six-degrees-of-freedom (6-DOF) unmanned aircraft system (UAS) and a double pendulum. We demonstrate that the lifted LPV models constitute a better approximation of the nonlinear system compared to linearized models, and the controllers designed based on these models outperform those designed using linearized models.

All in all, the main contribution of this work is threefold. First, we extend existing deep learning-based lifting linearization techniques to nonlinear systems with uncertain initial conditions. Specifically, we propose a learning-based approach for simultaneously characterizing the uncertain initial states in the lifted state space while learning a higher-dimensional linear approximation of the nonlinear dynamics.
We also extend our previous synthesis result \cite{farhood2021lpv} developed for NSLPV systems with uncertain initial conditions to accommodate a more appropriate characterization of the structured set of possible initial state values that arises in this work and ultimately design static and dynamic NSLPV controllers for the lifted LPV approximation of the nonlinear system; the static synthesis result is also specialized to the case of LTI systems.
Additionally, we present a robustness analysis framework based on IQC theory to evaluate the robust performance of the resulting closed-loop NSLPV systems, in which the lifted LPV plant model approximately captures the behavior of the original nonlinear system. The proposed analysis approach extends the results of our previous work \cite{farhood2024robustness} to handle signal IQCs and nonlinear uncertainties.

The rest of the paper is organized as follows. Section~\ref{sec_preliminaries} introduces the notation and provides a brief background on KOT. Section~\ref{sec_KoopmanModeling} details the lifted LPV modeling approach. The main NSLPV controller synthesis result, along with a complementary result specialized for LTI systems, is presented and proved in Section~\ref{sec_controldesign}. Section~\ref{sec_iqcAnalysis} presents the main IQC analysis result and the analysis-based controller tuning routine. Illustrative examples demonstrating the utility of the proposed approach are given in Section~\ref{sec_examples}. The limitations of the approach are discussed in Section~\ref{sec_limitations}, and concluding remarks are provided in Section~\ref{sec_conclusions}. Finally, an appendix is included, which provides an output feedback synthesis result for a more general class of NSLPV plants.

\section{Preliminaries}  \label{sec_preliminaries}

\subsection{Notation}
The sets of non-negative integers, positive real scalars, real vectors of dimension $n$, real matrices of size $n\times m$, and real $n\times n$ symmetric matrices are denoted by $\mathbb{Z}_+$, $\mathbb{R}_{++}$, $\mathbb{R}^n$, $\mathbb{R}^{n \times m}$, and $\mathbb{S}^n$, respectively. 
The set $\{1,\,2,\dots,\,p\}$ is denoted by $\mathbb{N}_p$. The set $\{0,\,1,\dots,\,p\}$ is denoted by $\mathbb{Z}_p$.
We write $X \succ 0 $ to indicate that the symmetric matrix $X$ is positive definite. The matrix elements which can be inferred from symmetry are denoted by $\ast$. The transpose of a matrix $X$ is denoted by $X^T$.
$I_n$ denotes an $n\times n$ identity matrix, $0_{n\times m}$ denotes an $n\times m$ zero matrix, and $\bm{1}_n$ denotes a vector of dimension $n$ whose components are all one. The subscripts are dropped when the dimensions $n$ and $m$ are clear from the context. 
For vectors $v\in\mathbb{R}^n$, the Euclidean norm is  $\lVert v \rVert_2 = (v^Tv)^{1/2}$. 
The block-diagonal augmentation of  $A_1,\, A_2,\, \dots,\, A_N$ is denoted by $\mathrm{diag}(A_1,\, A_2,\, \dots,\, A_N)$. Given  $x\in\mathbb{R}^n$, $\mathrm{diag}(x) = \mathrm{diag}(x_1,\, x_2,\, \dots,\, x_n)$, where $x_k \in \mathbb{R}$ is the $k^{th}$ element of $x$.
Given a positive definite matrix $P \in\mathbb{S}^m$, we define the ellipsoid ${\varepsilon}(P)=\{ y \in \mathbb{R}^m~|~ y^T P y \leq 1\}$.
$\{w_i\}_{i=p}^{q}$ denotes the sequence $(w_p,\,w_{p+1},\,\dots,\,w_q)$.
The Hilbert space $\ell_2^n$ consists of square-summable sequences $d = (d_0, \, d_1, \, \dots)$ having a finite $\ell_2^n$-norm,  defined as $\lVert d \rVert_{\ell_2^n}^2  = \sum_{k=0}^{\infty}d_k^Td_k$, where $d_k\in \mathbb{R}^{n}$ for all $k \in \mathbb{Z}_+$. When the dimension is clear from the context, we simply use the denotation $\ell_2$.
For Hilbert spaces $\mathcal{W}$ and $\mathcal{V}$, the $\mathcal{W}$-to-$\mathcal{V}$-induced norm of a bounded  operator $M$ mapping $\mathcal{W}$ to $\mathcal{V}$ is defined as $\lVert M \rVert_{\mathcal{W}\rightarrow \mathcal{V}} = \sup_{0\neq w\in \mathcal{W}} ({\lVert Mw \rVert_{\mathcal{V}}}/{\lVert w \rVert_{\mathcal{W}}})$. When $\mathcal{W}$ is a subset of $\mathcal{V}$, then $\lVert w \rVert_{\mathcal{W}} = \lVert w \rVert_{\mathcal{V}}$. The adjoint of a linear operator $\Pi$ is denoted by $\Pi^\ast$. 
A matrix sequence $P$ is $(h,\,q)$-eventually periodic, for some integers $h \geq 0$ and $q \geq 1$, if $P(h+k+iq) = P(h+k)$ for all $i,\,k \in \mathbb{Z}_+$. An LTV system is $(h,\,q)$-eventually periodic if all its state-space matrix sequences are $(h,\,q)$-eventually periodic. An LFT system $(M,\,\Delta)$ is $(h,\,q)$-eventually periodic if its nominal system $M$ is $(h,\,q)$-eventually periodic.

\subsection{Koopman Operator Theory}

This paper partly deals with generating lifted LPV models that approximate  nonlinear system dynamics over some desired operating envelopes. While this work does not directly focus on the Koopman operator, there are relevant works in the literature that are tied to KOT \cite{williams2015data, brunton2021modern, otto2021koopman, proctor2018generalizing}. For this reason, the key concepts of KOT  are briefly described in this section. Interested readers can refer to the works \cite{mauroy2020koopman,mezic2021koopman} to learn more about the history and state  of KOT in systems and control.

Consider an autonomous, nonlinear dynamical system $x_{k+1} = F(x_k)$, where $x_k \in \mathbb{X} \subseteq \mathbb{R}^n$ represents the state of the system at discrete-time step $k \in \mathbb{Z}_+$ and $F : \mathbb{X} \to \mathbb{X} $ is the nonlinear state transition map. The Koopman operator, $\mathcal{K}$, associated with the nonlinear dynamics governs the temporal evolution of the scalar-valued observable functions $\phi : \mathbb{X} \to \mathbb{R}$ belonging to $\mathbb{F}$, with $\mathbb{F}$ being some Koopman-invariant function space on the state space $\mathbb{X}$. Thus, for any observable $\phi \in \mathcal{F}$, the action of the Koopman operator $\mathcal{K} : \mathbb{F} \to \mathbb{F}$ on $\phi$ is defined as $\mathcal{K}\phi = \phi \circ F$ and at any time-instant $k$, we have $ \mathcal{K}\phi(x_k) =  \phi \left(F(x_k) \right) = \phi \left(x_{k+1}\right)$. The Koopman operator is linear even when the underlying dynamics defined by $F$ are nonlinear; however, the function space $\mathbb{F}$ is infinite dimensional. 
For practical purposes, a subspace $\mathbb{F}_{\mathcal{D}} \subset \mathbb{F}$ is considered, which is spanned by a finite dictionary of observable functions $\mathcal{D}  = \{ \phi_1,\, \phi_2,\, \dots,\, \phi_N \}$, typically with $N \gg n$. Finite dimensional Koopman-invariant subspaces have only been identified for a few dynamical systems \cite{brunton2016koopman}. In practice, data-driven approaches are commonly employed to approximate the infinite-dimensional Koopman operator, $\mathcal{K}$, with a finite-dimensional operator that acts on the subspace $\mathbb{F}_{\mathcal{D}}$. Typically, these methods seek a matrix $K \in \mathbb{R}^{N\times N}$ that approximates the nonlinear dynamics in the lifted observable space, as described by \( \Phi(x_{k+1}) \approx K \Phi(x_k) \), where \( \Phi(x) \coloneqq [\phi_1(x),\, \phi_2(x),\, \dots,\, \phi_N(x) ]^T \) represents the lifting function. 

KOT can also be extended to controlled nonlinear systems $x_{k+1} = F(x_k,\,u_k)$ through various approaches \cite{ korda2018linear, proctor2018generalizing, williams2016extending, iacob2024koopman}. Here, $u_k \in \mathbb{U} \subseteq \mathbb{R}^{n_u}$ represents the input at time $k$.
One such approach considers the controlled nonlinear system as an uncontrolled system evolving on the extended state space defined as the Cartesian product of the original state space and the space of all control sequences.
In this case, the observables $\psi : \mathbb{X} \times \mathbb{U} \to \mathbb{R}$ become functions of both the state and the input. A finite number of these observables can then be combined to construct a lifting function $\Psi(x,\,u) = [\psi_1(x,\,u),\, \psi_2(x,\,u),\, \dots,\, \psi_{\tilde{N}}(x,\,u) ]^T$ such that \( \Psi(x_{k+1},\,u_{k+1}) \approx K \Psi(x_k, \, u_k) \) for some $K \in \mathbb{R}^{{\tilde{N}}\times {\tilde{N}}}$. 
Typically, the lifting function is chosen such that the evolution of the system state in the lifted state space can be decoupled from that of the control input. For example, an LTI approximation \cite{korda2018linear,mamakoukas2019local} of the nonlinear system is obtained by imposing the following structure on the lifting function: $\Psi(x,\,u) = [\Phi(x)^T,\, u^T]^T$, with the function $\Phi$ mapping $\mathbb{R}^n$ to $\mathbb{R}^N$. Defining $\Psi(x,\,u) = [\Phi(x)^T,\, u^T,\, u^{[1]}\Phi(x)^T, \,  \ldots,\,u^{[m]}\Phi(x)^T]^T$ results in a more accurate bilinear approximation \cite{bruder2021advantages, folkestad2021koopman}, where $u^{[i]}$ denotes the $i^{th}$ element of $u$. A generalized LPV approximation \cite{iacob2024koopman, eyuboglu2024data} that encompasses both LTI and bilinear lifted models can be formulated by selecting a lifting function of the form $\Psi(x,\,u) = [\Phi(x)^T,\, u^T,\, \delta^{[1]}u^T, \,  \ldots,\,\delta^{[p]}u^T]^T$, where $\delta^{[i]}$, for $i \in \mathbb{N}_{p}$, are scheduling parameters defined as $\delta = [\delta^{[1]}, \, \ldots, \, \delta^{[p]}]^T = \mu\left( \Phi(x) \right)$, with $\mu : \mathbb{R}^N \to \mathbb{R}^p$ representing the scheduling parameter map. The LTI and bilinear models emerge as special cases when $\mu$ is chosen as the zero map and identity map, respectively.

\section{Lifted LPV Approximations of Nonlinear Systems}  \label{sec_KoopmanModeling}

We consider the following nonlinear dynamical system: 
\begin{equation}
    x_{k+1} = F(x_k,\,u_k,\,v_k), \quad x_0 \in \varepsilon(P),
    \label{eqn_NLeqn}
\end{equation}
where $x_k \in \mathbb{X} \subset \mathbb{R}^n$, $u_k \in \mathbb{U} \subset \mathbb{R}^{n_u}$, and $v_k \in \mathbb{V} \subset \mathbb{R}^{n_v}$ represent the state, control input, and process noise, respectively, with $\mathbb{X}$, $\mathbb{U}$, and $\mathbb{V}$ being bounded sets that define the operating envelope of interest. The system has an uncertain initial condition; $x_0$ lies in the ellipsoid ${\varepsilon}(P) \subseteq \mathbb{X}$, defined by a (known) positive definite matrix $P \in \mathbb{S}^n$. We assume that $(x_e,\,u_e,\,v_e) = (0,\,0,\,0)$ is an equilibrium of \eqref{eqn_NLeqn}, around which the system will be stabilized. This assumption is not limiting, as a non-zero equilibrium can be shifted to the origin through a change of variables.  Additionally, we assume that the system state variables are either directly measurable or estimated via an observer, with measurements and/or estimates corrupted by white Gaussian noise.

Using a data-driven lifting approach, we approximate the nonlinear dynamics \eqref{eqn_NLeqn} with a lifted LPV system of the form
\begin{equation}
z_{k+1} = A z_k + B_0 \tilde{u}_k  + \sum_{i=1}^{p} B_i \delta^{[i]}_k \tilde{u}_k,~~~z_0 = \Phi(x_0) \in \mathbb{I}, 
\label{eqn_Koopmansys}
\end{equation}
where $\tilde{u}_k = [u_k^T, ~ v_k^T]^T$, $z_k \in \mathbb{R}^N$ for $k \in \mathbb{Z}_+\backslash\{0\}$ approximates the lifted state vector $\Phi(x_k)$ in the observable space, $\delta_k = \mu(z_k)$ is the vector of scheduling parameters, and $\mathbb{I}$ denotes the set in which the initial lifted state vector resides. 
The approximate model (\ref{eqn_Koopmansys}) will be used to synthesize robust controllers for the original nonlinear system. 
The closed-loop guarantees established for this approximate model cannot be directly extended to the nonlinear system unless the approximation errors are incorporated into the control design process. Different metrics could be used to assess the significance of the approximation errors, including the loss function adopted in the learning process and the $\ell_2$-gain of the corresponding nonlinear error system, which can be approximately calculated based on a set of diverse inputs and initial conditions (obtained, for instance, using falsification as demonstrated in our previous work \cite{sinha2021lft}) and the associated outputs. IQC analysis, which is used in tuning the controllers in the proposed approach, can handle a wide range of uncertainties and, hence, allows incorporating the approximation errors into the control design process, as long as the discrepancies are modeled in a way amenable to IQC theory. However, this is not the focus of the present work.

For control design and analysis, we impose performance constraints on the original state of the system, which can be reconstructed as $x_k \approx \hat{x}_k = C z_k$, where $C \in \mathbb{R}^{n \times N}$ is a projection matrix. We embed the original state $x$ within the lifted state, i.e., $\Phi(x_k)$ is defined as $[x_k^T ~ \bar{\Phi}(x_k)^T]^T$, where $\bar{\Phi} : \mathbb{R}^n \to \mathbb{R}^{\bar{N}}$ with $\bar{N} = N-n$. Consequently, the solution for $C$ becomes trivial and is given by $C = \begin{bmatrix} I_n & 0 \end{bmatrix}$ \cite{narasingam2019koopman, folkestad2020extended, mamakoukas2019local, ma2019optimal}.
While embedding the original state within the lifted state can restrict the expressiveness of the model, it ensures that the mapping from the lifted state $\Phi(x_k)$ back to the original state $x_k$ is exactly linear. This not only simplifies both the learning and control design processes but also confines the modeling errors solely to the lifted state equation.

\subsection{Characterizing Uncertain Initial Conditions} \label{sec_UIC}

We characterize the set $\mathbb{I}$ in the lifted state space using ellipsoids. Since $z_0 = \Phi(x_0) = [x_0^T ~ \bar{\Phi}(x_0)^T]^T$ and $x_0$ is known to lie in $\varepsilon(P)$, we impose a separate constraint on the component $\bar{\Phi}(x_0)$, specifically constraining it to lie in the ellipsoid $\varepsilon(Q)$, defined by a positive definite matrix $Q \in \mathbb{S}^{\bar{N}}$. This allows for a potentially more accurate characterization of the set $\mathbb{I}$, as discussed in our previous work \cite{farhood2024robustness}. Thus, the set $\mathbb{I}$ is defined as
\begin{equation}
    \mathbb{I} = \{ (\bar{p},\,\bar{q}) \,\mid \, \bar{p} \in \varepsilon(P), ~ \bar{q} \in \varepsilon(Q) \}.
    \label{eqn_UICset}
\end{equation}

Ideally, ${\varepsilon}(Q) \supseteq \{ \bar{\Phi}(x_0) ~|~ x_0 \in {\varepsilon}(P) \}$, although equality is unlikely to hold as $\bar{\Phi}$ is nonlinear. For a given $\bar{\Phi}$, a data-driven approximation of $Q$ can be obtained by \rnum{1}) randomly sampling multiple vectors $\bar{p}_1,\,\dots,\,\bar{p}_{n_{\scriptstyle s}}$ from ${\varepsilon}(P)$; \rnum{2}) generating a set $\mathbb{Q} = \{\bar{q}_1,\,\dots,\,\bar{q}_{n_{\scriptstyle s}} \} \subset \mathbb{R}^{\bar{N}}$, where $\bar{q}_i = \bar{\Phi}(\bar{p}_i)$; and \rnum{3}) finding the minimum volume ellipsoid, centered at the origin, that contains the set $\mathbb{Q}$. The minimum volume ellipsoid can be computed by solving the following convex optimization problem \cite{boyd2004convex}:
\begin{equation}
\label{eqn_minVolell}
\begin{aligned}
\min_{(\bar{Q} \,\in\, \mathbb{S}^{\bar{N}})} \quad &  \log \det \bar{Q} ^{-1}\\
\text{subject to} \quad & \lVert \bar{Q}  \bar{q}_i \rVert_2 \leq 1, \quad i=1,\,\dots,\,n_{\scriptstyle s}. \\
\end{aligned}
\end{equation}
The constraint $\lVert \bar{Q}  \bar{q}\rVert_2 \leq 1$ is equivalent to $\bar{q} \in {\varepsilon}(Q)$, where $Q = \bar{Q}^T\bar{Q} = \bar{Q}^2$. Note that, when $ \{ \bar{\Phi}(x_0) ~|~ x_0 \in {\varepsilon}(P) \} \subset {\varepsilon}(Q) $, we introduce conservatism due to overapproximation. 
This overapproximation is a necessary trade-off to handle the nonlinear equality constraint $z_0 = [x_0^T ~ \bar{\Phi}(x_0)^T]^T$, while ensuring the computational tractability of the control synthesis and analysis problem. For certain lifting functions, such as $\bar{\Phi}(x)=x^2$, the extent of overapproximation can be significant, leading to notable conservatism. To mitigate this, we learn the lifting function $\bar{\Phi}$ and the corresponding ellipsoid ${\varepsilon}(Q)$ simultaneously. During learning, the volume of $\varepsilon(Q)$ is minimized, encouraging the choice of a lifting function that reduces the extent of overapproximation and the associated conservatism.

\subsection{Learning Approach}  \label{sec_learningApproach}

In this section, we present a data-driven methodology designed to simultaneously learn the lifted LPV model and characterize the set ${\varepsilon}(Q)$. We assume that multiple trajectories of \eqref{eqn_NLeqn}, each denoted as $\left(\{x_i\}_{i=0}^{N_s},\, \{u_i\}_{i=0}^{N_s-1}, \, \allowbreak \{v_i\}_{i=0}^{N_s-1} \right)$, are available within the envelope of interest and that they densely cover the envelope. The trajectories can be generated either through experiments or by simulating the (potentially unstable) nonlinear system in a closed-loop setting with a stabilizing (but not necessarily optimal) controller or in an open-loop setting with carefully designed inputs to ensure that the operating envelope of interest $(\mathbb{X},\, \mathbb{U},\, \mathbb{V})$ is not violated. For learning purposes, finite-horizon snapshots of each trajectory are considered since the approximate LPV model is unlikely to predict accurately the trajectories of a nonlinear system over a large horizon in an open-loop setting, especially when dealing with unstable systems. These snapshots, denoted $\left(\{x_{k+i}\}_{i=0}^{T},\, \{u_{k+i}\}_{i=0}^{T-1},\, \{v_{k+i}\}_{i=0}^{T-1} \right)$ for $k \in \mathbb{Z}_{N_s - T} $, form our training and validation datasets. 

We parameterize the function $\bar{\Phi}$ and the scheduling parameter map $\mu$ using multi-layer feed-forward neural networks. The temporal evolution in the lifted state space is modeled using a recurrent neural network (RNN). Our core RNN architecture is shown in Figure~\ref{fig_NNKoopman}, and is based on the ones given in \cite{otto2019linearly, iacob2021deep}. Given a trajectory from the training dataset, the initial lifted state is obtained using the lifting function as $z_k =\Phi(x_k)$. Over the prediction horizon, the lifted state $z_{k+i|k}$ is computed by recursively applying the LPV state equation \eqref{eqn_Koopmansys} with $z_k$ as the initial state value and $(\tilde{u}_{k},\,\tilde{u}_{k+1},\,\dots,\,\tilde{u}_{k+i-1})$ as the input. 
To simultaneously learn the maps $\bar{\Phi}$ and $\mu$ as well as the matrices $A,\, B_0, \, B_1,\dots, \, B_p, \, Q$, we formulate the following loss function, $\mathcal{L}$, which is minimized over all the trajectory snapshots in the training dataset:
\begin{equation}
    \mathcal{L} =  \mathcal{L}_{dyn} + \beta_1 \mathcal{L}_{ell} + \beta_2 \mathcal{L}_{vol}, \label{eqn_lossFunc}
\end{equation}
where the non-negative weights $\beta_1$ and $\beta_2$ are learning hyperparameters, along with the prediction horizon length $T$, and the loss functions $\mathcal{L}_{dyn}$, $\mathcal{L}_{ell}$, and $\mathcal{L}_{vol}$ are explained next.

\begin{figure}[t]
\centering
\includegraphics[width=0.6\textwidth]{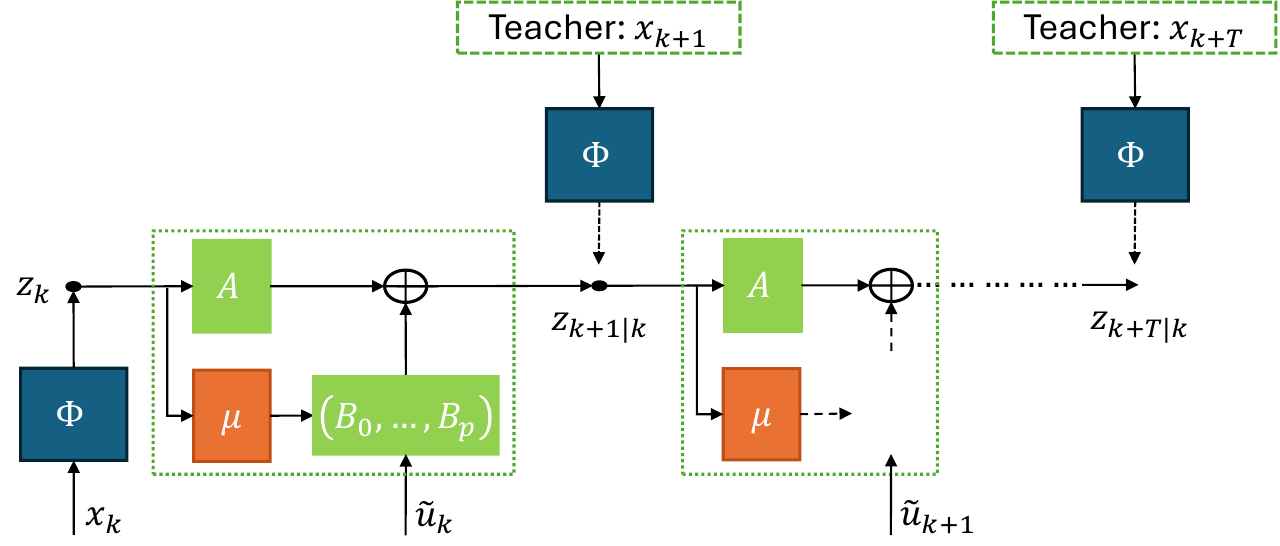} 
\caption{Lifted LPV model incorporated in an RNN architecture for learning.}
\label{fig_NNKoopman}
\end{figure}

To enforce that the nonlinear dynamics in the lifted state space evolve according to \eqref{eqn_Koopmansys}, we include the following term in the loss function:
\[ \mathcal{L}_{dyn} =  \frac{1}{T} \sum_{i=1}^{T} \rho^{i-1}  \lVert \Phi(x_{k+i}) - {z}_{k+i|k}\lVert_{\text{MSE}}, \]
where MSE refers to mean squared error and $\rho \in (0,\,1]$ is a decaying weight used to prioritize short-term predictions. 
The ellipsoid ${\varepsilon}(Q)$ must encompass the finite set $\mathbb{Q} =\left\{\bar{\Phi}(\bar{p}) ~|~ \bar{p} \in {\varepsilon}(P) \cap \{x_k,\,\dots,\,x_{k+T} \} \right\}$ corresponding to each trajectory in the training dataset. To simultaneously learn the ellipsoid, we need to incorporate the convex constraints outlined in \eqref{eqn_minVolell} into the learning problem. However, solving a constrained learning problem poses challenges, often necessitating iterative solutions and thereby substantially increasing the computational complexity \cite{pauli2021training}. To circumvent this issue, we relax the convex constraints and incorporate them directly into the loss function as follows:
\begin{equation*}
        \mathcal{L}_{ell} = \frac{1}{T+1}\sum_{i=0}^{T} \big(  \max\{0, \, 1 + \kappa - x_{k+i}^TPx_{k+i}\} \max\{0, \, \bar{\Phi}(x_{k+i})^T Q \bar{\Phi}(x_{k+i}) - 1\}  \big),
\end{equation*}
where $\kappa$ is a small positive constant. This loss function is added to minimize the weighted ``distance" of the elements of $\mathbb{Q}$ from ${\varepsilon}(Q)$. 

We have one more constraint, namely, $Q \succ 0$. In \cite{revay2020lipschitz, pauli2023lipschitz, wang2023direct, verhoek2023learning}, direct parametrizations for matrices are proposed to guarantee the satisfaction of specific LMIs. Similar to \cite{verhoek2023learning}, we use the spectral decomposition of a real symmetric matrix to parameterize the matrix $Q$ as
\begin{equation}
Q = V\Lambda V^T, ~\text{where}~ \Lambda = \mathrm{diag}\left(\exp{(\bar{d}_1)},\dots,\exp{(\bar{d}_{\bar{N}})}\right)~\text{and}~ V = \text{Cayley}(\mathtt{V}).
    \label{eqn_pdQ}
\end{equation}
Here, the vector $\bar{d} = (\bar{d}_1,\dots,\bar{d}_{\bar{N}}) \in \mathbb{R}^{\bar{N}}$ and the skew-symmetric matrix $\mathtt{V} = -\mathtt{V}^T \in \mathbb{R}^{\bar{N} \times \bar{N}}$ serve as free variables. The skew-symmetric matrix $\mathtt{V}$ can alternatively be parameterized as $\mathtt{U} - \mathtt{U}^T$, with $\mathtt{U} \in \mathbb{R}^{\bar{N} \times \bar{N}}$ as the free variable. The Cayley transform of the skew-symmetric matrix $\mathtt{V}$, defined as $\text{Cayley}(\mathtt{V}) = (I - \mathtt{V})(I + \mathtt{V})^{-1}$, is employed to parameterize the orthogonal matrix $V$ \cite{trockman2020orthogonalizing}. As $V$ is orthogonal, the eigenvalues of $Q$ correspond to the diagonal elements of $\Lambda$, all of which are positive, thereby establishing the positive definiteness of $Q$. To minimize the volume of the ellipsoid ${\varepsilon}(Q)$, we incorporate the following into the loss function:
\[ \mathcal{L}_{vol} = {\exp \left(- \frac{1}{2} \sum_{i=1}^{\bar{N}} \bar{d}_i \right)}.\]
The volume of the ellipsoid ${\varepsilon}(Q)$ is proportional to $\sqrt{ \det{Q^{-1}} }$,  which is equivalent to $\mathcal{L}_{vol}$ when $Q$ is parameterized using \eqref{eqn_pdQ}. This relationship is evident by recognizing that the determinant of $Q^{-1} = V\Lambda^{-1}V^T$ is equal to the product of its eigenvalues $\exp{(-\bar{d}_1)},\dots,\exp{(-\bar{d}_{\bar{N}})}$.
\vskip 0.1in
\begin{prop}
If $\mathcal{L}_{ell} = 0$, then $ \mathbb{Q} \subseteq {\varepsilon}(Q)$ is guaranteed.
\end{prop}
\begin{proof}
The summand in $\mathcal{L}_{ell}$ is non-zero if and only if $x_{k+i} \in \tilde{{\varepsilon}}(P) \wedge \bar{\Phi}(x_{k+i}) \notin {{\varepsilon}}(Q)$, where $\tilde{{\varepsilon}}(P) = \{r ~|~ r^TPr \leq 1 + \kappa \}$ is an expansion of ${\varepsilon}(P)$ by a margin $\kappa > 0$. Therefore, if $\mathcal{L}_{ell} = 0$, $\neg \left( x_{k+i} \in \tilde{{\varepsilon}}(P) \wedge \bar{\Phi}(x_{k+i}) \notin {{\varepsilon}}(Q) \right)$ holds for all $i \in \mathbb{Z}_T$. This is equivalent to $ x_{k+i} \in \tilde{{\varepsilon}}(P) \rightarrow \bar{\Phi}(x_{k+i}) \in {{\varepsilon}}(Q) $ for all $i \in  \mathbb{Z}_T$. Additionally, $ x_{k+i} \in {{\varepsilon}}(P) \rightarrow x_{k+i} \in \tilde{{\varepsilon}}(P)$ for all $ i \in \mathbb{Z}_T$, as ${{\varepsilon}}(P) \subset \tilde{{\varepsilon}}(P)$. By the transitivity of implication, we get $ x_{k+i} \in {{\varepsilon}}(P) \rightarrow \bar{\Phi}(x_{k+i}) \in {{\varepsilon}}(Q)$ for all $i \in \mathbb{Z}_T$. Thus, all elements of the set $\mathbb{Q}$ belong to $ {\varepsilon}(Q)$, ensuring $ \mathbb{Q} \subseteq {\varepsilon}(Q)$.
\end{proof}

The learning hyperparameters $\beta_1$ and $\beta_2$ can be tuned to ensure that $\mathcal{L}_{ell} = 0$ holds. In fact, when $\beta_2=0$, the training process can trivially achieve $\mathcal{L}_{ell} = 0$ by selecting an excessively large ellipsoid $\varepsilon(Q)$. While this choice guarantees constraint satisfaction, it introduces significant conservatism due to overapproximation. Consequently, one must balance $\beta_1$ and  $\beta_2$ by performing a trade-off study between the size of the learned ellipsoid and the extent of possible initial state values it covers. A smaller value of $\beta_2$ places less emphasis on minimizing the size of the ellipsoid, making it easier for the training process to find an ellipsoid that contains all initial state samples from the learning dataset. On the other hand, a larger $\beta_2$ results in tighter ellipsoids but may lead to constraint violations or degraded model accuracy.

\subsection{Refining the Learned Ellipsoid} 

The learning approach does not guarantee convergence to an ellipsoid that fully encompasses the set of possible initial state values derived from the training and validation datasets, and even if it does, the resulting ellipsoid may not be of minimum volume. Therefore, after learning, the matrix $Q$ is recomputed by solving problem \eqref{eqn_minVolell} using all possible initial state values derived from the learning dataset. 
This data-driven estimation of $Q$ may still be inaccurate if the nonlinear system trajectories used for learning do not densely cover the set $\varepsilon(P)$. Although we can generate trajectories by appropriately sampling initial states from $\varepsilon(P)$, the coverage may still be insufficient due to the limited number of trajectories. Hence, to improve this data-driven estimation, we employ an iterative method, where we search for a counterexample $x_c$ such that $ x_c \in {\varepsilon}(P) \wedge  \bar{\Phi}(x_c) \notin {\varepsilon}(Q) $ holds true. If a counterexample is found, $\bar{\Phi}(x_c)$ is added to the dataset, and problem \eqref{eqn_minVolell} is solved again. This iterative process continues until no counterexample can be found. If no counterexample exists, we can conclude that the ellipsoid ${\varepsilon}(Q)$ accurately captures the set of initial state values. However, finding a counterexample or proving the nonexistence of one is challenging, as the formula to be satisfied is nonlinear. While satisfiability modulo theories  (SMT) solvers such as \texttt{dReal} \cite{gao2013dreal} can be applied to this problem, they are typically not scalable. 
In this work, we adopt a heuristic approach by solving a constrained nonlinear optimization problem using MATLAB’s \texttt{fmincon}, where the constraint is $x_0 \in \varepsilon(P)$ and the objective is to maximize the utility function $\bar{\Phi}(x_0)^T Q \bar{\Phi}(x_0)$. A solution $x_0$ is considered a counterexample if the optimal cost is greater than one.

\section{Robust Control Design} \label{sec_controldesign}

In this section, we present a synthesis approach for designing static full-state feedback NSLPV controllers for the LPV system described in \eqref{eqn_Koopmansys}. Without loss of generality, we assume that the scheduling parameters satisfy $\rvert \delta^{[i]}_k \lvert \leq 1$ for all $k \in \mathbb{Z}_+$  and $i \in \mathbb{N}_p$. The upper and lower bounds for each scheduling parameter can be estimated using samples of the original state $x$ from the learning dataset, along with the learned functions $\Phi$ and $\mu$. Given these bounds, we can introduce new ``normalized" parameters satisfying the unit bound constraints in the LPV system, and adjust the relevant matrices accordingly to account for this normalization. Under this assumption, let $\bm{\delta}$ represent the set of all permissible trajectories of the scheduling parameter $\delta$. Then, the synthesis goal is to design a controller such that the closed-loop LPV system is robustly stable and, for some $\gamma > 0$,
\begin{equation}
    \sup \, \{\lVert e \rVert_{\ell_2} \,\mid\, \lVert d \rVert_{\ell_2} \leq 1,~ \delta \in \bm{\delta},~ z_0 \in \mathbb{I} \}< \gamma
    \label{eqn_perfinequality}
\end{equation}
holds, where $e$ represents the performance output, $d \in \ell_2$ denotes the exogenous disturbance consisting of process noise  and measurement noise, and the set $\mathbb{I}$ is defined as in \eqref{eqn_UICset}. The assumption that the $\ell_2$-norm of $d$ is bounded by one is made without loss of generality since any non-unit bound can be absorbed into the system matrices.

For control design, we define the performance output as 
\begin{equation}
    e = \mathrm{diag}(c)[\hat{x}^T ~ u^T]^T = C_e z + D_e u, 
    \label{eqn_perfOut}
\end{equation}
where $c \in \mathbb{R}^{n_e}$, with $n_e = n+n_u$, is a vector of penalty weights which must be chosen judiciously by the designer to optimize some performance measure. Since $\hat{x} = Cz$, we get $C_e = \mathrm{diag}(c) [ C^T ~ 0_{N\times n_u} ]^T$ and $D_e = \mathrm{diag}(c) [ 0_{n_u\times n} ~ I_{n_u} ]^T$. The full-state feedback controller will be implemented with $y_k = \Phi(\tilde{x}_k)$ as the measurement output, where $\tilde{x}_k$ represents the true state ${x}_k$ corrupted by noise. Then, we can approximate the measurement output as 
\begin{equation}
    y_k \approx \Phi(x_k) + W w_k \approx z_k + W w_k = \hat{y}_k, 
    \label{eqn_measOut}
\end{equation}
where $W \in \mathbb{R}^{N \times n_w}$ is a weighting matrix derived from the standard deviations of the measurement noise in the lifted state space, with $w$ representing unit-variance white Gaussian noise. Given the lifting function $\Phi$ and the characteristics of the measurement noise in the original state space, the noise characteristics in the lifted state space can be estimated using data.

\begin{figure}[t]
\centering
\includegraphics[width=0.65\textwidth]{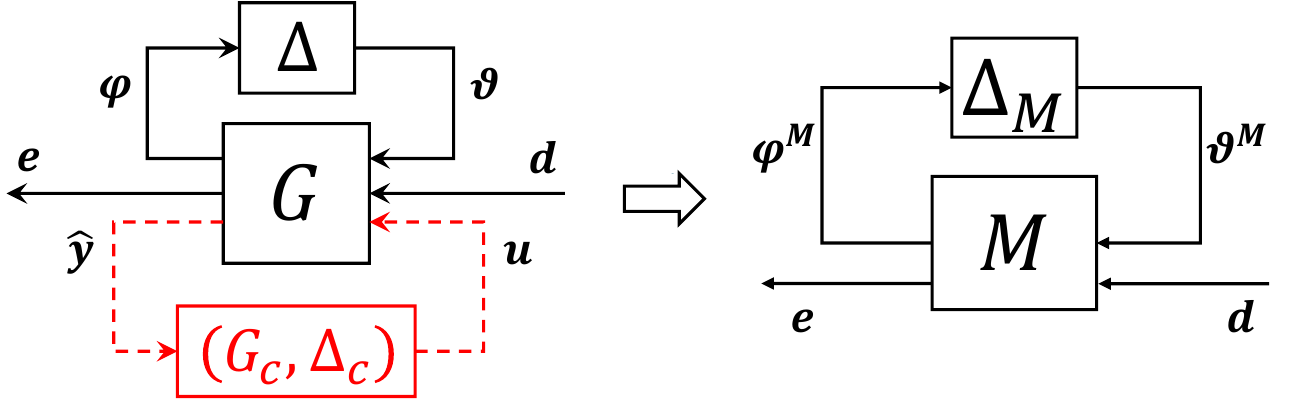} 
\caption{LFT representation of the LPV model.}
\label{fig_LFTKoopman}
\end{figure}

The LPV system can be equivalently represented by the interconnection, $(G,\,{\Delta})$, of a nominal LTI system $G$ and a perturbation operator $\Delta$, as shown in Figure~\ref{fig_LFTKoopman}, where 
$$\Delta(k) = \mathrm{diag}(\delta_k^{[1]} I_{m_1}, \, \delta_k^{[2]} I_{m_2}, \, \dots, \, \delta_k^{[p]} I_{m_p})$$
and the nominal system $G$ is defined as follows: 
\begin{equation}
    \begin{bmatrix} z_{k+1} \\ \varphi_k \\ e_k \\ \hat{y}_k \end{bmatrix} = 
    \begin{bmatrix} A_{ss} & A_{sp} & B_{1s} & B_{2s}\\ A_{ps} & A_{pp} & B_{1p} & B_{2p}\\ C_{1s} & C_{1p} & D_{11} & D_{12} \\ C_{2s} & C_{2p} & D_{21} & 0 \end{bmatrix} 
    \begin{bmatrix} z_k \\ \vartheta_k \\ d_k \\ u_k  \end{bmatrix}, \quad \vartheta_k = \Delta(k)\varphi_k, \quad z_0 \in \mathbb{I}.
    \label{eqn_LFTkoop}
\end{equation}
Here, $d_k = [w_k^T, ~ v_k^T]^T \in \mathbb{R}^{n_d}$, with $n_d=n_w+n_v$, and $\vartheta_k,\, \varphi_k \in \mathbb{R}^m$, where $m = \sum_{i=1}^{p}m_i$ and $m_i = n_u + n_v$ for all $i \in \mathbb{N}_p$. The state-space matrices of $G$ are defined as follows:
\begin{equation}
    \begin{split}
    A_{ss} &= A,\quad A_{sp} =  \begin{bmatrix}{B}_1 & {B}_2 & \dots & {B}_p\end{bmatrix},\quad 
    B_{1s} = \begin{bmatrix}0_{N \times n_w} & B_{02}\end{bmatrix}, \quad B_{1p} = \bm{1}_p \otimes \mathrm{diag}(0_{n_u \times n_w} ,\, I_{n_v}), \\
    B_{2s} &= B_{01}, \quad B_{2p} = \bm{1}_p \otimes \begin{bmatrix}I_{n_u} & 0_{{n_u} \times n_v}\end{bmatrix}^T,\quad  C_{1s} = C_e,\quad  D_{12} = D_e,\quad  C_{2s} = I,\quad  D_{21} = \begin{bmatrix}W & 0_{N \times n_v}\end{bmatrix},
    \end{split}
    \label{eqn_LFTssmatrices}
\end{equation}
with $\otimes$ denoting the Kronecker product and all other matrices being zero. The matrices $B_{01}$ and $B_{02}$ are such that $B_0 = \begin{bmatrix}B_{01} & B_{02}\end{bmatrix}$.
For $\delta \in \bm{\delta}$, let $\bm{\Delta}$ be the set of all possible perturbation operators $\Delta$. Then, we define $(G,\,\bm{\Delta}) = \{(G,\,{\Delta}) ~|~ \Delta \in \bm{\Delta} \}$ as an uncertain LFT system that encapsulates the behavior of the LPV system.

\subsection{Controller Synthesis Results}

\begin{lemma} \label{lemma1}
    Given matrices $X \in \mathbb{R}^{\bar{n} \times \bar{m}}$ and $Y \in \mathbb{R}^{\bar{m} \times \bar{m}}$, the following statements are equivalent:
    \begin{enumerate}[(i)]
        \item There exist positive definite matrices $R, S \in \mathbb{S}^{\bar{n}}$ such that
            \begin{equation}
                \begin{bmatrix}  R & I \\ I & S \end{bmatrix} \succeq 0, \quad X^T S X \prec Y; \label{eqn_lemma1a}
            \end{equation}
        \item There exists a positive definite matrix $R \in \mathbb{S}^{\bar{n}}$ such that
            \begin{equation}
                \begin{bmatrix} Y & X^T \\ X & R     \end{bmatrix} \succ 0. \label{eqn_lemma1b}
            \end{equation}
    \end{enumerate}
\end{lemma} \vskip 0.05in
\begin{proof}
    Applying the Schur complement formula to the first inequality in \eqref{eqn_lemma1a}, we get $R^{-1} \preceq S$, which, together with the second inequality in \eqref{eqn_lemma1a}, leads to $ X^T R^{-1} X \prec Y$. Applying the Schur complement formula to the preceding inequality gives \eqref{eqn_lemma1b}. This proves (i) implies (ii). The converse direction, (ii) implies (i), can be shown by simply choosing $S = R^{-1}$ and applying Schur complement formula to \eqref{eqn_lemma1b}.
\end{proof} \vskip 0.1in

\begin{thm} \label{col1}
    Consider an LPV system with an LFT representation defined as in \eqref{eqn_LFTkoop}. Suppose that the state variables are exactly measurable, i.e., $C_{2s} = I$, $C_{2p} = 0$, and $D_{21} = 0$, then a static, state-feedback, $(\bar{h},\,1)$-eventually periodic NSLPV synthesis, with $\bar{h} \geq 0$, ensuring the validity of the performance inequality in \eqref{eqn_perfinequality} exists if there exist positive definite matrices $R_0(k) \in \mathbb{S}^{N}$, $R_i(k) \in \mathbb{S}^{m_i}$,  $S_i(k) \in \mathbb{S}^{m_i}$  for $i \in \mathbb{N}_p$ and $k = 0,\,1,\,\dots,\,\bar{h}$, and positive scalars $b$, $f_{11}$, $f_{12}$, $f_{2}$, $g$, and $t$ such that
    \begin{align}
        & b + f_{11} + f_{12} + f_2 < 2\gamma, \quad  \begin{bmatrix} F_1 & \Gamma^T \\ \Gamma & R_0(0) \end{bmatrix} \succ 0,\label{eqn_LPVSFsynthesis1} \\[4pt]
        & N_R^T \left\{ H \begin{bmatrix} R_0(k) & 0 & 0 \\ 0 & \bar{R}(k) & 0 \\ 0 & 0 & g I \end{bmatrix} H^T -  
            \begin{bmatrix} R_0(k+1) & 0 & 0 \\ 0 & \bar{R}(k) & 0 \\ 0 & 0 & b I \end{bmatrix} \right\} N_R \prec 0,\label{eqn_LPVSFsynthesis2}  \\[4pt]
        & \begin{bmatrix}
            \begin{bmatrix} \bar{S}(k) & 0 \\ 0 & f_2 I \end{bmatrix} - H_{22}^T\begin{bmatrix} \bar{S}(k) & 0 \\ 0 & tI \end{bmatrix}H_{22} & H_{12}^T \\ H_{12} & R_0(k+1)
        \end{bmatrix} \succ 0, \label{eqn_LPVSFsynthesis3}  \\[4pt]
        & \begin{bmatrix} R_i(k) & I \\ I & S_i(k) \end{bmatrix} \succeq 0, \quad  \begin{bmatrix} g & 1 \\ 1 & f_2 \end{bmatrix} \succeq 0, \quad  \begin{bmatrix} t & 1 \\ 1 & b \end{bmatrix} \succeq 0, \label{eqn_LPVSFsynthesis4} 
    \end{align}
    for $i = 1,\dots,p$ and $k = 0,1,\dots,\bar{h}$, where $R_0(\bar{h}+1) = R_0(\bar{h})$, $\Gamma = \mathrm{diag}(P^{-1/2}, \, Q^{-1/2})$, $F_1 = \mathrm{diag}(f_{11} I_n, \, f_{12}I_{\bar{N}})$,
    \begin{equation*}
    \begin{split}
        \bar{R}(k) &= \mathrm{diag}(R_1(k),\, R_2(k),\, \dots,\, R_p(k)),  \\[2pt]
        \bar{S}(k) &= \mathrm{diag}(S_1(k),\, S_2(k),\, \dots,\, S_p(k)),  \\[2pt]
        \mathrm{Im}\, N_R &= \mathrm{Ker} \begin{bmatrix} B_{2s}^T & B_{2p}^T & D_{12}^T \end{bmatrix}, ~~  N_R^TN_R = I, 
    \end{split} \qquad
        H = \begin{bmatrix}
            H_{11} & H_{12} \\ H_{21} & H_{22} \end{bmatrix} =
            \left[ \begin{array}{c|cc}
            A_{ss} & A_{sp} & B_{1s} \\ \midrule A_{ps} & A_{pp} & B_{1p} \\ C_{1s} & C_{1p} & D_{11} 
        \end{array} \right],
    \end{equation*} 
    with $\mathrm{Im}\, Z$ and $\mathrm{Ker}\,Z$ denoting the image and kernel of a matrix $Z$, respectively.
\end{thm} \vskip 0.05in
\begin{proof}
    The conditions in the theorem statement are derived by specializing Theorem~\ref{thm1}, given in Appendix~\ref{app_A}, to the case of $(0,\,1)$-eventually periodic LPV plants, i.e., standard (stationary) LPV plants, with exactly measurable states. Theorem~\ref{thm1} itself is closely related to the results in our previous work \cite{farhood2021lpv}.
    The initial state $z_0 \in \mathbb{I}$ of the lifted model can be expressed as $z_0 = \Gamma \xi$, where $\xi = (\xi_1, \, \xi_2)$, with $\xi_1 \in \mathbb{R}^n$ and $\xi_2 \in \mathbb{R}^{\bar{N}}$ satisfying $\lVert \xi_1 \rVert_2 \leq 1$, $\lVert \xi_2 \rVert_2 \leq 1$, and
    \begin{equation}
        \Gamma = \mathrm{diag}(P^{-\frac{1}{2}}, \, Q^{-\frac{1}{2}}). \label{eqn_UIC}
    \end{equation}
    Since $C_{2s} = I$, $C_{2p} = 0$, and $D_{21} = 0$, we get $N_S^T = [0_{(m+n_d)\times N} ~ I_{m+n_d}]$, and so, for $k = 0,\,1,\,\dots,\,\bar{h}$, Eq.~\eqref{eqn_thm1_BSC} in Theorem~\ref{thm1} simplifies to 
    \begin{equation*}
             H_{12}^TS_0(k+1)H_{12}  \prec \mathrm{diag}\left(\bar{S}(k),\,f_2I \right) - H_{22}^T \mathrm{diag}\left(\bar{S}(k),\,tI \right)H_{22}.       
    \end{equation*}
    The preceding inequality, together with the first inequality in \eqref{eqn_thm1_coupling} corresponding to $i = 0$ and $k=1,\,\dots,\,\bar{h}+1$, are equivalent to \eqref{eqn_LPVSFsynthesis3} for $k = 0,\,1,\,\dots,\,\bar{h}$ by Lemma~\ref{lemma1}. Similarly, the second inequality in \eqref{eqn_thm1_UIC}, together with the first inequality in \eqref{eqn_thm1_coupling} corresponding to $i = 0$ and $k=0$, are equivalent to the second inequality in \eqref{eqn_LPVSFsynthesis1}. The remaining conditions in Theorem~\ref{thm1} are  included as is in this theorem.
\end{proof} \vskip 0.05in

If the conditions in Theorem~\ref{col1} are feasible, an $(\bar{h},\,1)$-eventually periodic static controller $(G_c,\,\Delta_c)$, defined as
\begin{equation}
    \begin{bmatrix} \varphi^c_k \\ u_k  \end{bmatrix} = 
    \begin{bmatrix}   A^c_{pp}(k) & B^c_{p}(k) \\  C^c_{p}(k) & D^c(k) \end{bmatrix} 
    \begin{bmatrix} \vartheta^c_k \\ \hat{y}_k \end{bmatrix}, ~~ \vartheta^c_k = \Delta_c(k)\varphi^c_k \in \mathbb{R}^{m^c(k)}, 
    \label{eqn_LFTcon}
\end{equation}
can be constructed, where $\Delta_c(k) = \mathrm{diag}(\delta_k^{[1]} I_{m^c_1(k)}, \, \dots, \, \delta_k^{[p]} I_{m^c_p(k)})$, $m^c_i(k) = \text{rank}(R_i(k) - S_i^{-1}(k)) \leq m_i$, and $m^c(k) = \sum_{i=1}^{p}m^c_i(k)$ for all $i \in \mathbb{N}_p$ and $k \in \mathbb{Z}_+$. 
From \eqref{eqn_LFTcon}, we obtain $\varphi^c_k = A^c_{pp}(k)\Delta_c(k)\varphi^c_k  + B^c_{p}(k)\hat{y}_k$, which, after some algebraic manipulation, yields 
\[ \vartheta^c_k = \Delta_c(k)\varphi^c_k = \Delta_c(k)\left(I - A^c_{pp}(k)\Delta_c(k)\right)^{-1}B^c_{p}(k) \hat{y}_k.\]
The control law for the nonlinear system \eqref{eqn_NLeqn}, assuming no measurement noise, i.e., $\hat{y}_k = z_k$, and $z_k = \Phi(x_k)$, can then be written as 
\begin{equation*}
  u_k = C^c_{p}(k)\vartheta^c_k + D^c(k)\hat{y}_k = \left( D^c(k) + C^c_{p}(k)\Delta_c(k)(I - A^c_{pp}(k)\Delta_c(k))^{-1}B^c_{p}(k) \right)\Phi(x_k), \quad \delta_k = \mu\left(\Phi(x_k)\right).
\end{equation*}
Note that the controllers are synthesized for an LPV system with normalized scheduling parameters. Therefore, during controller implementation, the parameter values $\delta_k = \mu(\Phi(x_k))$ must be normalized by applying the same upper and lower parameter bounds that were used to ``normalize'' the LPV plant model.

In Theorem~\ref{col1}, the length $\bar{h}$ of the finite-horizon component of the controller is a design parameter. When the system has nonzero initial conditions in addition to exogenous disturbances, choosing $\bar{h} > 0$ can potentially enhance closed-loop performance by allowing the controller to initially adopt strategies that prioritize mitigating the effects of the uncertain initial state \cite{farhood2021lpv,farhood2008control}. Note that by increasing the finite horizon length of the desired controller, we are enlarging the solution space of the synthesis optimization problem, which will lead to either improved results or the same ones as before.

The controller matrices $A^c_{pp}(k)$, $B^c_{p}(k)$, $C^c_{p}(k)$, and $D^c(k)$ are obtained by solving the following LMI for $k = 0,1,\dots,\bar{h}$:
\begin{equation}
    \mathbf{H} + \mathbf{Q}^T\mathbf{J}(k)^T\mathbf{P} + \mathbf{P}^T\mathbf{J}(k)\mathbf{Q} \prec 0,  \quad \text{where}
    \label{eqn_kypLMI}
\end{equation}
\begin{equation*} 
    \begin{split}
    \mathbf{P} &= \begin{bmatrix} 0_{m^c \times N}  & 0_{m^c \times m}  & I_{m^c} & 0_{m^c \times (N + m + m^c)} & 0_{m^c \times n_d} & 0_{m^c \times n_e}  \\ B_{2s}^T & B_{2p}^T & 0_{n_u \times m^c}  & 0_{n_u \times (N + m + m^c)} & 0_{n_u \times n_d} & b^{-1/2}D_{12}^T \end{bmatrix}\!, \\[8pt]
    \mathbf{Q} &= \begin{bmatrix} 0_{m^c \times (N + m + m^c)} & 0_{m^c \times N} & 0_{m^c \times m} & I_{m^c} & 0_{m^c \times n_d} & 0_{m^c \times n_e} \\ 
    0_{N \times (N + m + m^c)} & I_N & 0_{N \times m} &  0_{N \times m^c} &  0_{N \times n_d} &  0_{N \times n_e}\end{bmatrix}\!, \\[8pt]
    \mathbf{H} &= \begin{bmatrix} - Y_{11} & -Y_{12} & \mathbf{A} & 0_{(N+m) \times m^c} & \mathbf{B} & 0_{(N+m) \times n_e} \\
                                      \ast & -I_{m^c} & 0_{m^c \times (N+m)} & 0_{m^c \times m^c} & 0_{m^c \times n_d} & 0_{m^c \times n_e} \\
                                       \ast & \ast & -X_{11} & -X_{12} & 0_{(N+m)\times n_d} & \mathbf{C}^T \\
                                       \ast & \ast & \ast & -X_{22} & 0_{m^c \times n_d} & 0_{m^c \times n_e} \\
                                       \ast & \ast & \ast & \ast & -I_{n_d} & (bf_2)^{-1/2}D_{11}^T \\
                                        \ast & \ast & \ast & \ast & \ast & -I_{n_e}
    \end{bmatrix}\!, 
    \end{split} 
\end{equation*}
\begin{equation*}
    \mathbf{J}(k) = \begin{bmatrix} A^c_{pp}(k) & B^c_{p}(k) \\ C^c_{p}(k) & D^c(k) \end{bmatrix}, \quad
     \mathbf{A} = \begin{bmatrix} A_{ss} & A_{sp} \\ A_{ps} & A_{pp} \end{bmatrix}, \quad \mathbf{B} = f_2^{-1/2}\begin{bmatrix} B_{1s} \\ B_{1p}  \end{bmatrix},  \quad  X_{12} = -\begin{bmatrix} 0_{N \times m^c} \\ \bar{S}(k)\bar{E} \end{bmatrix}, \quad Y_{12} = \begin{bmatrix} 0_{N \times m^c} \\ \bar{E} \end{bmatrix}, 
\end{equation*}
$\mathbf{C} = b^{-1/2}\begin{bmatrix} C_{1s} & C_{1p} \end{bmatrix}$, $Y_{11} = \mathrm{diag}\left(R_0(k+1), \, \bar{R}(k)\right)$, $X_{11} = \mathrm{diag}\left(R_0(k)^{-1}, \, \bar{S}(k)\right)$, $X_{22} = I + \bar{E}^T\bar{S}(k)\bar{E}$, and $\bar{E} \in \mathbb{R}^{m \times m^c}$ such that  $\bar{E} \bar{E}^T = \bar{R}(k) - \bar{S}(k)^{-1}$. Here, the time dependence of $\mathbf{H}$, $X_{11}$, $X_{12}$, $X_{22}$, $Y_{11}$, $Y_{12}$, $\bar{E}$, and $m^c$ has been suppressed for convenience.
These conditions are derived by specializing the procedure outlined in our previous work \cite{farhood2021lpv}, which builds on another earlier work \cite{farhood2012nonstationary}, to the static state-feedback controller case. This involves setting $C_{2s} = I$, $C_{2p} = 0$, $D_{21} = 0$, and $S_0(k) = R_0(k)^{-1}$ (which follows from Lemma~\ref{lemma1}) in our previous formulations \cite{farhood2021lpv,farhood2012nonstationary}.
We can take $\bar{E}(k) = (\bar{R}(k) - \bar{S}(k)^{-1})^{1/2}$ when $m_i^c(k) = m_i$ for all $i \in \mathbb{N}_p$. This can also be ensured by changing the first inequality in \eqref{eqn_LPVSFsynthesis4} to a strict inequality. 

\vskip 0.1in
\begin{corollary} \label{col2}
Suppose the lifted model defined in \eqref{eqn_Koopmansys} is LTI, i.e., $B_i = 0$ for $i \neq 0$. Then, a static state-feedback LTI synthesis ensuring the validity of the inequality in \eqref{eqn_perfinequality} exists if there exist a positive definite matrix $R \in \mathbb{S}^{N}$, a matrix $S \in \mathbb{R}^{n_u \times N}$, and positive scalars $b$, $f_{11}$, $f_{12}$, and $f_2$ such that
\begin{equation}
    \begin{bmatrix}
        - R & A_{ss} R + B_{2s} S & B_{1s} & 0 \\ \ast & - R & 0 & (C_{1s} R + D_{12} S)^T \\ \ast & \ast & -f_2 I & D_{11}^T \\ \ast & \ast & \ast & -b I
    \end{bmatrix} \prec 0, \quad b + f_{11} + f_{12} + f_2 < 2\gamma, \quad  \begin{bmatrix} F_1 & \Gamma^T \\ \Gamma & R \end{bmatrix} \succ 0, \label{eqn_LTISFsynthesis1}
\end{equation}
where $\Gamma$ and $F_1$ are defined in Theorem~\ref{col1}, and the state-space matrices are defined in \eqref{eqn_LFTssmatrices}. If the above problem is feasible, then the control law for the nonlinear system \eqref{eqn_NLeqn}, assuming no measurement noise and $z_k = \Phi(x_k)$, can be written as $u = SR^{-1}\Phi(x)$. 
\end{corollary} \vskip 0.05in
\begin{proof}
    When $B_i = 0$ for $i \neq 0$, the LFT representation in \eqref{eqn_LFTkoop} reduces to an LTI system. Correspondingly, the controller in \eqref{eqn_LFTcon} simplifies to $u_k = D^c(k)\hat{y}_k$. Under these conditions, the LMI in \eqref{eqn_kypLMI} takes the following form:
    \begin{equation}
        \begin{bmatrix}
        - R_0(k+1) & A_{ss} + B_{2s}D^c(k) & f_2^{-1/2}B_{1s} & 0 \\ \ast & - R_0(k)^{-1} & 0 & b^{-1/2}(C_{1s} + D_{12} D^c(k))^T \\ \ast & \ast & -I & (bf_2)^{-1/2}D_{11}^T \\ \ast & \ast & \ast & I
    \end{bmatrix} \prec 0.
    \label{eqn_LTISFsynthesis2}
    \end{equation}
    This inequality is obtained by replacing the matrices $A_{sp}$, $A_{ps}$, $A_{pp}$, $B_{1p}$, $B_{2p}$, $C_{1p}$, $\bar{R}(k)$, $\bar{S}(k)$, and $\bar{E}(k)$ in \eqref{eqn_kypLMI} with empty matrices having at least one dimension equal to zero, as the dimensions $m$ and $m^c$ become zero for an LTI plant and an LTI or LTV controller, respectively. Rather than solving for $R_0$ separately using Theorem~\ref{col1} and then, if successful, computing $D^c$, we can lump these two steps together. That is, we can solve for $R_0$ and $D^c$ simultaneously by replacing the inequalities \eqref{eqn_LPVSFsynthesis2}-\eqref{eqn_LPVSFsynthesis4} with \eqref{eqn_LTISFsynthesis2} in Theorem~\ref{col1}. For the synthesis of a time-invariant controller, we set $\bar{h} = 0$. The conditions in the corollary statement are then obtained from the modified theorem by \rnum{1}) defining $R_0(0) = R$; \rnum{2}) pre- and post-multiplying \eqref{eqn_LTISFsynthesis2} by $\mathrm{diag}(I,\, R,\, \sqrt{f_2}I, \, \sqrt{b}I)$; and \rnum{3}) setting $D^c(0)R = S$ in the resulting inequality. 
\end{proof}

\section{IQC Analysis}  \label{sec_iqcAnalysis}

In the previous section, we presented state-feedback controller synthesis results under the assumption that the system states are exactly measurable. While this assumption simplifies the controller's structure, enabling the synthesis of static controllers instead of dynamic ones, it is often impractical. Consequently, it is essential to account for the measurement noise in any analysis associated with the synthesized static controllers to ensure that the closed-loop performance  remains satisfactory even when the state measurements are corrupted with noise. Moreover, during synthesis, the scheduling parameters are treated as static LTV (SLTV) uncertainties that can vary arbitrarily fast. However, in practice, these parameters, which depend on the system state, usually exhibit bounded rates of change. The process noise inputs are also modeled as $\ell_2$ signals during synthesis, which encompass a broad range of disturbances, some of which may not be seen in real-world conditions. These factors motivate the need for an analysis approach that addresses the aforementioned limitations to yield more realistic and less conservative performance guarantees for the LPV system.

\subsection{Analysis Result}

Here, we present an IQC-based analysis approach that extends the result of our previous work \cite{farhood2024robustness} to accommodate signal IQCs and nonlinear uncertainties in the analysis. The analysis result is developed for an $(h_M,\,q_M)$-eventually periodic uncertain system  $(M,\,\bm{\Delta}_M) = \{ (M,\,\Delta_M) ~|~ \Delta_M \in \bm{\Delta}_M \}$ having the following space-space representation: 
\begin{equation}
    \begin{bmatrix} x^M_{k+1} \\ \varphi^M_k \\ e_k  \end{bmatrix} = 
    \begin{bmatrix} A_{M}(k) & B_{M_1}(k) & B_{M_2}(k) \\ C_{M_1}(k) & D_{M_{11}}(k) & D_{M_{12}}(k) \\C_{M_2}(k) & D_{M_{21}}(k) & D_{M_{22}}(k) \end{bmatrix} 
    \begin{bmatrix} x^M_k \\ \vartheta^M_k \\ d_k \end{bmatrix}, \quad \vartheta^M_k = \Delta_M(k) \varphi^M_k,  
    \label{eqn_CLLFT}
\end{equation}
where $x^M_k \in R^{n_M(k)}$, $d_k  \in R^{n_d(k)}$, and $e_k  \in R^{n_e(k)}$ denote the state, disturbance input, and performance output of the system at discrete-time instant $k$, respectively.

Let $\Pi$ be a self-adjoint operator factored as $\Psi^{\ast}J\Psi$, where $J= [J_{ij}]_{i,j=1,2}$, $ J_{ij} = \mathrm{diag}(J_{ij}(0),J_{ij}(1),\ldots)$ is a bounded block-diagonal (memoryless) operator, and $\Psi$ is a linear, bounded, causal operator that can be generally represented by an asymptotically stable, discrete-time, LTV system. The set $\bm{\Delta}_M$ satisfies the IQC defined by $\Pi$ (denoted $\bm{\Delta}_M \in \text{IQC}(\Pi)$) if for $x^{\Psi}_0=0$ and all $\varphi^M \in \ell_2$, $\vartheta^M = \Delta_M\varphi^M$, and $\Delta_M \in \bm{\Delta}_M$, the following condition holds:
\begin{equation}
\sum_{k=0}^{\infty} r_k^TJ(k)r_k \geq 0, \quad \text{where}
\label{eqn_multdelta}
\end{equation}
\begin{equation*}
    \begin{bmatrix} x^\Psi_{k+1} \\ r^{[1]}_{k} \rule{0mm}{4mm} \\ r^{[2]}_{k} \rule{0mm}{4mm}  \end{bmatrix} = 
    \begin{bmatrix} A_{\Psi}(k) & B_{\Psi_1}(k) & B_{\Psi_2}(k) \\ C_{\Psi_1}(k) & D_{\Psi_{11}}(k) & D_{\Psi_{12}}(k) \\C_{\Psi_2}(k) & D_{\Psi_{21}}(k) & D_{\Psi_{22}}(k) \end{bmatrix} \begin{bmatrix} x^{\Psi}_k\\ \varphi^M_k \\ \vartheta^M_k  \end{bmatrix},  
\end{equation*}
$r_k = (r^{[1]}_k,\, r^{[2]}_k)$, $(A_\Psi(k),\,B_{\Psi_1}(k),\, \dots, \,D_{\Psi_{22}}(k))$ are the matrices defining the state-space realization of $\Psi$, and $x^{\Psi}_k \in \mathbb{R}^{n_\Psi(k)}$ denotes the state of the system $\Psi$.  Note that the matrix $J(k)=[J_{ij}(k)]_{i,j=1,2}$ is partitioned conformably with the partitioning of $r_k=(r^{[1]}_k,\, r^{[2]}_k)$. The operator $\Pi = [\Pi_{ij}]_{i,j=1,2}$ is a positive-negative multiplier \cite{HuLacerdaSeiler2017} if $\ip{\mu}{\Pi_{11}\mu}_{\ell_2}\geq \epsilon \ip{\mu}{\mu}_{\ell_2}$ and $\ip{\nu}{\Pi_{22}\nu}_{\ell_2}\leq -\epsilon \ip{\nu}{\nu}_{\ell_2}$ for some $\epsilon\geq 0$ and all $\mu,\nu\in\ell_2$, where $\ip{\cdot}{\cdot}_{\ell_2}$ denotes the inner product on $\ell_2$.

The set $\mathbb{D} \subseteq \ell_2$ satisfies the signal IQC defined by the self-adjoint operator $\Omega = \Theta^\ast U \Theta$ (denoted $\mathbb{D} \in \text{SigIQC}(\Omega)$), where $U$ and $\Theta$ are defined similarly to $J$ and $\Psi$, respectively, if for $x^\Theta_0=0$ and all $d \in \mathbb{D}$, the following holds:
\begin{equation}
\sum_{k=0}^{\infty} l_k^T U(k) l_k \geq 0, \quad \text{where}
\label{eqn_multdis}
\end{equation}
\vspace{-1mm}
\begin{equation*}
    \begin{bmatrix} x^{\Theta}_{k+1} \\  l_k  \end{bmatrix} = 
    \begin{bmatrix}  A_{\Theta}(k) & B_{\Theta}(k)\\C_{\Theta}(k) & D_{\Theta}(k)  \end{bmatrix} 
    \begin{bmatrix} x^{\Theta}_k  \\  d_k \end{bmatrix}.  
\end{equation*}
Given systems $M$, $\Psi$, and $\Theta$, an augmented system $L$ can be formed with state $x^L = (x^M,\, x^\Psi,\, x^\Theta)$, input $(\vartheta^M,\, d)$, output $(r^{[1]},\, e,\, r^{[2]},\, d,\, l)$, and state-space realization given by $(A_L(k),\, B_L(k),\, C_L(k),\, D_L(k))$.
Using \eqref{eqn_CLLFT}, the state-space equation in \eqref{eqn_multdelta} can be rewritten as
\begin{equation}
    \begin{bmatrix} x^\Psi_{k+1} \\ r^{[1]}_{k} \rule{0mm}{4mm} \\ r^{[2]}_{k} \rule{0mm}{4mm}  \end{bmatrix} = 
    \begin{bmatrix} A_{\Psi}(k) & B_{\Psi_1}(k) & B_{\Psi_2}(k) \\ C_{\Psi_1}(k) & D_{\Psi_{11}}(k) & D_{\Psi_{12}}(k) \\C_{\Psi_2}(k) & D_{\Psi_{21}}(k) & D_{\Psi_{22}}(k) \end{bmatrix} 
    \begin{bmatrix} 0 & I & 0 & 0 & 0 \\ C_{M_1}(k) & 0 & 0 & D_{M_{11}}(k) & D_{M_{12}}(k) \\ 0 & 0 & 0 & I & 0\end{bmatrix} \begin{bmatrix} x^L_k  \\ \vartheta^M_k \\ d_k \end{bmatrix}, 
    \label{eqn_augss1}
\end{equation}
where $x^L_k=\begin{bmatrix}(x^M_k)^T& (x^\Psi_k)^T& (x^\Theta_k)^T\end{bmatrix}^T$. Similarly, the equations in \eqref{eqn_CLLFT} and \eqref{eqn_multdis} can be expressed as
\begin{equation}
\begin{split}
    \begin{bmatrix} x^M_{k+1}  \\ e_k  \end{bmatrix} &= 
    \begin{bmatrix} A_{M}(k) & B_{M_1}(k) & B_{M_2}(k) \\C_{M_2}(k) & D_{M_{21}}(k) & D_{M_{22}}(k) \end{bmatrix} 
    \begin{bmatrix} I & 0 & 0 & 0 & 0 \\ 0 & 0 & 0 & I & 0 \\ 0 & 0 & 0 & 0 & I \end{bmatrix}
    \begin{bmatrix} x^L_k  \\ \vartheta^M_k \\ d_k \end{bmatrix}, \\
    \begin{bmatrix} x^{\Theta}_{k+1} \\  l_k  \end{bmatrix} &= 
    \begin{bmatrix}  A_{\Theta}(k) & B_{\Theta}(k)\\C_{\Theta}(k) & D_{\Theta}(k)  \end{bmatrix} 
    \begin{bmatrix} 0 & 0 & I & 0 & 0 \\ 0 & 0 & 0 & 0 & I \end{bmatrix}
    \begin{bmatrix} x^L_k  \\ \vartheta^M_k \\ d_k \end{bmatrix}. 
    \end{split}
    \label{eqn_augss2}
\end{equation}
Finally, from \eqref{eqn_augss1} and \eqref{eqn_augss2}, the state-space equation of the augmented system $L$ can be obtained easily through straightforward matrix manipulations with the state-space matrices given by
\begin{align*}
    A_L(k) &= \begin{bmatrix}A_M(k) & 0 & 0 \\ B_{\Psi_1}(k)C_{M_1}(k) & A_{\Psi}(k) & 0 \\ 0 & 0 & A_{\Theta}(k)   \end{bmatrix}, \hspace{9mm}
    B_L(k) = \begin{bmatrix}B_{M_1}(k) & B_{M_2}(k) \\ B_{\Psi_1}(k)D_{M_{11}}(k) + B_{\Psi_2}(k) &   B_{\Psi_1}(k)D_{M_{12}}(k) \\
    0 & B_{\Theta}(k) \end{bmatrix}, \\[2pt]
    C_L(k) &= \begin{bmatrix} D_{\Psi_{11}}(k)C_{M_1}(k) & C_{\Psi_1}(k) & 0 \\ C_{M_2}(k) & 0 & 0 \\ D_{\Psi_{21}}(k)C_{M_1}(k) & C_{\Psi_2}(k) & 0 \\ 0 & 0 & 0 \\ 0 & 0 & C_{\Theta}(k) \end{bmatrix}, \quad
    D_L(k) = \begin{bmatrix} D_{\Psi_{11}}(k) D_{M_{11}}(k) + D_{\Psi_{12}}(k) & D_{\Psi_{11}}(k) D_{M_{12}}(k) \\
            D_{M_{21}}(k) & D_{M_{22}}(k) \\
            D_{\Psi_{21}}(k) D_{M_{11}}(k) + D_{\Psi_{22}}(k) & D_{\Psi_{21}}(k) D_{M_{12}}(k) \\
            0 & I \\ 0 & D_{\Theta}(k)
            \end{bmatrix}.
\end{align*}

\vskip 0.1in
\begin{thm} \label{thm2}
Consider the uncertain system $(M,\,\bm{\Delta}_M)$ described in \eqref{eqn_CLLFT}. Let $d \in \mathbb{D} \subseteq \ell_2$, and suppose the initial state $x^M_0$ of the system is uncertain and can be expressed as $x^M_0 = \Gamma \xi$, where $\Gamma \in \mathbb{R}^{n_M(0) \times s}$ and $\xi = (\xi_1, \, \dots, \, \xi_a)$, with $\xi_i \in \mathbb{R}^{s_i} $, $\lVert \xi_i \rVert_2 \leq 1$, and $\sum_{i=1}^{a} s_i = s \leq n_M(0)$. Then, this uncertain system has a robust performance level of $\gamma$, i.e.,
\begin{equation*}
    \sup \, \{\lVert e \rVert_{\ell_2} \, \mid \,\, \lVert \xi_1 \rVert_{2} \leq 1,\, \dots,\, \lVert \xi_a \rVert_{2} \leq 1, 
     \lVert d \rVert_{\ell_2} \leq 1,\, d \in \mathbb{D},\, \Delta_M \in \bm{\Delta}_M \} < \gamma, \quad \text{if}
\end{equation*}
\begin{enumerate}[(a)]
    \item $(M,\,\bm{\Delta}_M)$ is well-posed;
    \item there exists a positive-negative multiplier $\Pi = \Psi^\ast J \Psi$ with $(h_\Pi,\,q_\Pi)$-eventually periodic factors $\Psi$ and $J$ such that $\bm{\Delta}_M \in \text{IQC}(\Pi)$;
    \item there exists a signal IQC multiplier $\Omega = \Theta^\ast U \Theta$ with $(h_\Omega,\,q_\Omega)$-eventually periodic factors $\Theta$ and $U$ such that $\mathbb{D} \in \text{SigIQC}(\Omega)$; and
    \item there exist positive scalars $t$, $f_{11}$, $f_{12}$, \dots, $f_{1a}$, $f_2$, $g$ and a sequence $\bar{X}(k) = [\bar{X}_{ij}(k)]_{i,j=1,2,3} \in \mathbb{S}^{n_M(k)+n_\Psi(k)+n_\Theta(k)}$ for $k = 0,\,1,\,\dots,\,h+q$, where $h = \max{(h_M,\,h_\Pi,\,h_\Omega)}$, $q$ is the least common multiple of $q_M$, $q_\Pi$, and $q_\Omega$, $\bar{X}_{11}(k) \in \mathbb{S}^{n_M(k)}$, and $\bar{X}(h+q) = \bar{X}(h)$, satisfying the following LMIs:
\begin{align}
  & t + \sum_{i=1}^{a} f_{1i} + f_2 < 2\gamma, \quad \Gamma^T \bar{X}_{11}(0) \Gamma \prec F_1, \quad \begin{bmatrix} g & 1 \\ 1 & t \end{bmatrix} \succeq 0, \label{eqn_IQCanalysis1}
\\
 & \begin{bmatrix} I & 0 \\ A_L(k) & B_L(k) \\ C_L(k) & D_L(k) \end{bmatrix}^T \begin{bmatrix}-\bar{X}(k)&0&0\\0& \bar{X}(k+1)&0\\0&0&\tilde{J}(k) \end{bmatrix} \begin{bmatrix} I & 0 \\ A_L(k) & B_L(k) \\ C_L(k) & D_L(k) \end{bmatrix}  \prec 0, \label{eqn_IQCanalysis2}
\end{align}

for $k = 0,\,1,\,\dots,\,h+q-1$, with $F_1 = \mathrm{diag}(f_{11} I_{s_1},\, f_{12} I_{s_2},\, \dots,\, f_{1a} I_{s_a})$ and
\begin{equation*}
     \tilde{J}(k) = \begin{bmatrix} J_{11}(k) & 0 & J_{12}(k) & 0 & 0 \\ 0 & g I_{n_e(k)} & 0 & 0 & 0 \\ J_{12}^T(k) & 0 & J_{22}(k) &0 &0\\
    0 & 0 & 0 & -f_2I_{n_d(k)} & 0 \\    0 & 0 & 0 & 0 & U(k) \end{bmatrix}. 
\end{equation*}
\end{enumerate}
\end{thm} \vskip 0.05in
\begin{proof}
Pre- and post-multiply inequality \eqref{eqn_IQCanalysis2} by $\begin{bmatrix}  (x^L_k)^T & (\vartheta^M_k)^T & (d_k)^T \end{bmatrix}$ and its transpose to get
\begin{equation*}
(x^L_{k+1})^T \bar{X}(k+1)x^L_{k+1} - (x^L_{k})^T \bar{X}(k) x^L_k + r_k^TJ(k)r_k + g e_k^Te_k - f_2 d_k^Td_k + l_k^TU(k)l_k<0.
\end{equation*}
The conditions of Theorem \ref{thm2} ensure the robust stability of the uncertain system; specifically, in the absence of a disturbance input $d$, condition \eqref{eqn_IQCanalysis2} reduces to the linear operator inequality in Theorem~1 of \cite{fry2021robustness}, which, along with items $(a)$ and $(b)$ in the above theorem statement, implies that the uncertain system is robustly stable. Summing both sides of the above inequality from $k=0$ to $k=\infty$ and using the fact that the uncertain system is robustly stable,~we~get
\begin{equation}
\sum_{k=0}^{\infty}r_k^TJ(k)r_k + \sum_{k=0}^{\infty}l_k^TU(k)l_k + g \lVert e\rVert_{\ell_2}^2< f_2\lVert d\rVert_{\ell_2}^2 + (x^L_0)^T\bar{X}(0)x^L_0.\label{eq:iqcproof}
\end{equation}
Note that $(x^L_0)^T\bar{X}(0)x^L_0=(x^M_0)^T\bar{X}_{11}(0)x^M_0$, as the IQC filters, $\Psi$ and $\Theta$, have zero initial states. Since $x^M_0 = \Gamma \xi$, we can write $(x^L_0)^T\bar{X}(0)x^L_0 = \xi^T\Gamma^T\bar{X}_{11}(0)\Gamma \xi$, which is less than $\xi^T F_1 \xi$ for $\xi \neq 0$ from the second inequality in \eqref{eqn_IQCanalysis1}. Also, $\xi^T F_1 \xi = \sum_{i=1}^{a} f_{1i} \lVert \xi_i \rVert_2^2 \leq \sum_{i=1}^{a} f_{1i}$  since $\lVert \xi_i \rVert_2 \leq 1$ for $i=1,\dots,a$. Thus,  we have $(x^L_0)^T\bar{X}(0)x^L_0 < \sum_{i=1}^a f_{1i}$; this inequality also holds for $\xi = 0$ since $f_{1i}$, for $i=1,\dots,a$, are positive.

From (\ref{eq:iqcproof}), observing that the first two terms in this inequality  are nonnegative, $\lVert d\rVert_{\ell_2}\leq 1$, $(x^L_0)^T\bar{X}(0)x^L_0 < \sum_{i=1}^a f_{1i}$, and  $g\geq t^{-1}$ (which follows from the third inequality in  \eqref{eqn_IQCanalysis1}), we obtain
\begin{equation}
    \lVert e\rVert_{\ell_2}^2< t \left(\sum_{k=1}^{a}f_{1i} + f_2 \right).
    \label{eqn_thm2proof}
\end{equation}
From the first inequality in \eqref{eqn_IQCanalysis1}, we  have $\sum_{i=1}^{a}f_{1i} + f_2<2\gamma-t$; then, the inequality \eqref{eqn_thm2proof} can be written as  $\lVert e\rVert_{\ell_2}^2 < 2\gamma t-t^2=\tilde{\gamma}$. Clearly, we would like to choose the value of $t>0$ that results in the minimum value of $\gamma$ for a given $\tilde{\gamma}$. That is, we would like to find the value of $t$ that minimizes the function $\gamma=f(t)=(\tilde{\gamma}+t^2)/(2t)$, which is convex on $\mathbb{R}_{++}$. It is not difficult to see that the optimal point is $t^*=\sqrt{\tilde{\gamma}}$. For this value of $t$, $\gamma=\sqrt{\tilde{\gamma}}$, which leads to the inequality $\lVert e\rVert_{\ell_2}<\gamma$.
\end{proof}

In addition to being a generalization of the main result of our prior work \cite{farhood2024robustness}, Theorem \ref{thm2} can be viewed as an extension of a counterpart result given in \cite{fry2021robustness}. Specifically, Theorem \ref{thm2} can accommodate multiple initial conditions to better characterize the uncertain initial state; for instance, $a$ vectors, for some positive integer $a$, composed of different subsets of the uncertain initial state variables can be restricted to lie in different ellipsoids. If we have only one condition on the uncertain initial state, i.e., the vector composed of all the uncertain initial state variables is restricted to lie in an ellipsoid, then this result becomes merely a variant of the one provided in \cite{fry2021robustness}.  While it is true that enabling a better characterization of the uncertain initial state introduces additional variables into the analysis problem, namely, $a-1$ variables, these variables are scalars and the associated modifications in the LMIs do not alter the size of the constraints. Thus, the added computational complexity will not be significant in general. As for the conservativeness of the result, it turns out that the finer the partitioning of the vector of uncertain initial state variables is, which corresponds to a larger value of $a$ and, hence, more conditions on the uncertain initial state variables, the more conservative the result becomes. In our prior work \cite{farhood2024robustness}, a different approach is used to prove the less general version of this result, where only linear perturbations are considered and signal IQCs are not included. The argument used in the proof therein, which involves solving a square $\ell_2$ problem, touches on the added conservativeness stemming from increasing the number of these partitions or, as viewed in that work, ``inputs'' to the system. Typically, an analysis optimization problem is solved, which involves minimizing the robust performance level $\gamma$. Since the nominal system is assumed to be stable, then the major factor for the feasibility of this problem would be the uncertainties, specifically, the number, size, and type of the different uncertainties and their associated bounds. For instance, an infeasible problem could be rendered feasible by reducing the bounds on the uncertainties. The choice of the IQC multipliers plays a crucial role as well, and that is why in the analysis optimization problem we jointly solve for appropriate multipliers from prespecified sets of suitable multipliers for the various uncertainties and disturbance inputs that result in the least conservative analysis outcomes afforded by the approach.

\subsection{Controller Analysis and Tuning} \label{sec_iqctuning}

The feedback interconnection of the controller $(G_c,\,\Delta_c)$ from \eqref{eqn_LFTcon} with the plant $(G,\,\Delta)$ from \eqref{eqn_LFTkoop}, as shown in Figure~\ref{fig_LFTKoopman}, results in a closed-loop $(h_M,\,q_M)$-eventually periodic LFT representation $(M,\,\Delta_M)$. Here, $h_M = \bar{h}$, $q_M = 1$, and $\Delta_M$ varies within a predefined set $\bm{\Delta}_M$, with $\Delta_M(k) = \mathrm{diag}(\delta_k^{[1]} I_{m_1 + m^c_1(k)}, \, \dots, \, \delta_k^{[p]} I_{m_p + m^c_p(k)})$. The steps involved in obtaining the state-space matrices of $M$ from $G$ and $G_c$ are straightforward and, hence, omitted for brevity. The state of $M$ is identical to the state of $G$ because $G_c$ is static; hence, the initial state $x^M_0$ can be expressed as $\Gamma \xi$, where $\Gamma$ is defined in \eqref{eqn_UIC}. Theorem~\ref{thm2} can now be used to compute the robust performance level of the uncertain system $(M,\,\bm{\Delta}_M)$, which approximately captures the behavior of the controlled nonlinear system. 

During  analysis, we characterize the scheduling parameters as rate-bounded SLTV (RB-SLTV) uncertainties. The rate bounds can be estimated based on the parameter trajectories generated using the functions $\Phi$ and $\mu$ from the original state trajectories in the learning dataset. 
A discrete-time, time-varying parameterization of the IQC multipliers provided in \cite{veenman2016robust} for rate-bounded time-varying parametric uncertainties is used to characterize the RB-SLTV uncertainties; see \cite{fry2021robustness} for details.
As in \cite{palframan2017robustness}, we characterize the measurement noise using IQC multipliers for ``banded white" signals from \cite{jonsson1999iqc}. A frequency band of $[-\pi, ~\pi]$ is used to indicate that the noise signals maintain a constant power spectral density across the entire frequency range, similar to white noise. The process noise inputs can be treated as $\ell_2$ signals or characterized using signal IQCs if suitable multipliers can be formulated.

The IQC-based analysis result can also be used to tune the static state-feedback NSLPV controllers for the nonlinear system. We employ a minimization routine based on IQC analysis for this purpose \cite{palframan2017robustness,Fry2019,sinha2025control}. The controller $(G_c,\,\Delta_c)$ is synthesized based on a modified plant $(\bar{G},\,\Delta)$ using Theorem~\ref{col1}, where $\bar{G}$ is derived from $G$ by removing the measurement noise channels. The nominal systems $G$ and $\bar{G}$ also differ in their performance outputs. While the performance output of $G$ can be any linear function of state and input variables of interest as specified by the designer, the performance output of $\bar{G}$ is defined as in \eqref{eqn_perfOut}, with the vector of penalty weights $c$ treated as an optimization variable. Adjustments to these weights alter the nominal system $\bar{G}$ and, consequently, the synthesized controller. In this work, we adopt the IQC-based tuning routine from \cite{Fry2019} to optimize the penalty weights, ensuring that the resulting closed-loop LFT system $(M,\,\bm{\Delta}_M)$ achieves an optimal (albeit locally) robust performance level. Starting with an initial selection of the penalty weight vector $c$, this routine employs a Hessian-based minimization approach combined with a line search strategy to iteratively refine the penalty weights until no significant improvement in the robust performance level $\gamma$ is observed. First, the gradient of $\gamma$ with respect to $c$ is numerically estimated using finite differences, where each component $c_i$ of $c$ is perturbed by $\delta_c$ and the resulting change in $\gamma$ is recorded as the $i^{th}$ element of the gradient vector. The Hessian of $\gamma$ is updated using the Broyden–Fletcher–Goldfarb–Shanno (BFGS) update rule \cite{dennis1996numerical}. Using these quantities, the descent direction $\Vec{p}$ is determined, and a line search within a predefined interval $[0,\,\alpha_{max}]$ is conducted to find the optimal step size $\alpha$ along the descent direction for updating $c$; and so, the updated value of $c$ would be $c + \alpha\Vec{p}$.

\begin{remark} \label{rem1}
    An LTI controller can also be designed for the nonlinear system using the tuning routine. In this case, the synthesis problem will be solved using Corollary~\ref{col2}, based on a nominal LTI system $G_{nom}$. This nominal system can be derived from $(\bar{G},\,\Delta)$ by setting $\Delta = 0$ or through Jacobian linearization of the nonlinear system. Although the controller is designed using an LTI model that may not adequately capture the nonlinear system's dynamics, the analysis guiding the control design process can be conducted using the lifted LPV model, which provides a more accurate representation of the nonlinear system's behavior. 
\end{remark}

Theorem~\ref{thm2} is used to compute the robust performance level of the uncertain system $(M,\,\bm{\Delta}_M)$, which approximately captures the behavior of the controlled nonlinear system. Ideally, we would like to account for the discrepancies between the uncertain model and the original system outputs during the analysis to generate performance guarantees that extend to the nonlinear system. These discrepancies can be treated as uncertainties and modeled in a way amenable to IQC theory, which can handle a wide range of uncertainties. The challenge, however, lies in capturing modeling discrepancies as uncertainties without introducing excessive conservatism in the analysis problem. In our approach, we strive to obtain a lifted LPV approximation of the nonlinear system such that the analysis results based on this lifted model are qualitatively indicative of the behavior of the controlled nonlinear system. For controller tuning based on IQC analysis, such qualitative measures are often adequate, as the objective is not to validate any guarantees at this stage but rather to refine a suboptimal initial controller and ultimately generate a tuned controller with a substantially improved robust performance level.

In our previous work, \cite{sinha2021lft} we have characterized and quantified approximation errors within an IQC framework. Specifically, we modeled the discrepancies between a nonlinear system and its LPV approximation as a norm-bounded, causal, dynamic uncertainty, i.e., a dynamic uncertainty was added to the LPV system to account for the modeling errors. The norm bound on the uncertainty ($\ell_2$-gain of the nonlinear error system) was then estimated using falsification. \cite{annpureddy2011s} However, this approach is not applicable when the nonlinear error system is unstable, as the corresponding $\ell_2$-gain will not be finite. 
Recently, Eyuboglu et al. \cite{eyuboglu2024koopman} incorporated an additive dynamic uncertainty into the Koopman-based lifted model to account for the approximation error under the assumption that the dynamic uncertainty is bounded. This assumption holds for any open-loop stable nonlinear system, provided its lifted approximation also remains stable.
To validate a given controller, the dynamic uncertainty can be added to the controlled system $(M,\,\bm{\Delta}_M)$; see, for example, the case study in our previous work. \cite{sinha2021lft} However, including  an additive uncertainty in $(M,\,\bm{\Delta}_M)$ does not affect the qualitative comparison of robust performance levels across candidate controllers, rendering it ineffective for the controller tuning routine.

The approximation error can also be incorporated as an additive residual signal in the state equation. \cite{strasser2023robust, mamakoukas2022robust}
This signal can be interpreted as a disturbance that is bounded pointwise in time, and the uncertain system can then be analyzed against this disturbance using IQC-based analysis tools. \cite{sinha2025control} However, this characterization tends to produce conservative results in the present setting, as it assumes that the components of the disturbance signal vary independently within fixed upper and lower bounds. Further work is therefore needed to either refine IQC-based characterizations of residual signals or develop alternative approaches for incorporating uncertainty directly into the model.

\section{Illustrative Examples}  \label{sec_examples}

In this section, we apply the proposed approach to design static full-state feedback NSLPV controllers for nonlinear systems based on their lifted LPV approximations. As examples, we consider continuous-time nonlinear dynamics of a 6-DOF UAS and a double pendulum. Discrete-time nonlinear simulations are performed in MATLAB using \texttt{ode23} with some sampling time $\Delta t$. The learning problem is solved using the PyTorch-based deepSI toolbox \cite{beintema2021nonlinear}. The convex optimization problem \eqref{eqn_minVolell} is solved using the CVX toolbox \cite{Grant2014} in MATLAB. The LMIs in the controller synthesis and analysis problems are solved using YALMIP/MOSEK \cite{lofberg2004yalmip, aps2019mosek}. For analysis, RB-SLTV IQC multipliers \cite{veenman2016robust} are selected with a basis function of length $2$ and  poles equal to $0.5$. For the ``banded’’ white signal IQC multipliers \cite{jonsson1999iqc,palframan2017robustness}, a basis function of length $1$ is chosen, and the pole location is set to $0.95$. All the computations are done on a desktop with $32$ GB of RAM and an Intel Xeon E-2224G $3.5$ GHz CPU ($4$ cores). 

In both examples, the functions $\bar{\Phi} : \mathbb{R}^n \to \mathbb{R}^{\bar{N}}$ and $\mu : \mathbb{R}^N \to \mathbb{R}^p$ are parameterized using deep neural networks with two hidden layers, each consisting of $64$ neurons and employing the ELU activation function. For learning purposes, we generate a total of $7000$ trajectories, using $5000$ for training and the remaining $2000$ for validation. The training is done using the Adam optimizer over $1000$ epochs with a batch size of $512$. The dimension $N$ of the lifted state space, number of scheduling parameters $p$, sampling time $\Delta t$, and prediction horizon length $T$ for both examples are given in Table~\ref{tbl_LPVmodels}. When selecting $N$ and $p$, there is a trade-off between model accuracy and complexity. Increasing these values can result in a more accurate LPV approximation but also renders the resulting IQC analysis problem more computationally intensive. The prediction horizon length $T$ significantly impacts the stabilizability of the learned LPV models; specifically, we observed that smaller values of $T$ result in unstabilizable LPV models, while increasing $T$ can effectively circumvent this problem. Through several trials, the loss function weights $\beta_1$, $\beta_2$, and $\rho$ are set to $10^{-4}$, $1$ and $0.9$, respectively. The weights $\beta_1$ and $\beta_2$ are chosen with the objective of learning a compact ellipsoid $\varepsilon(Q)$ that includes the possible initial values of the lifted state without significantly compromising the accuracy of the lifted~model.


\subsection{6-DOF UAS}

The dynamics of the UAS are described by the following nonlinear differential equations:
\begin{equation}
\label{eqn_UASeom}
\dot{\omega} =\texttt{I}^{-1}\big( \bm{m}(u, \textbf{v}_r, \omega)-{\omega}\times \texttt{I}{\omega} \big), \quad \dot{{\bf v}} =  \texttt{m}^{-1}\bm{f}(u, \textbf{v}_r, \omega) + \texttt{g} -  {\omega}\times{\bf v}, \quad \dot{{\lambda}} = E(\phi,\theta)\omega, \quad  \dot{\textbf{p}} = R_{b}^{I}(\lambda)\textbf{v},
\end{equation}
where ${\lambda}= [\phi,\, \theta,\, \psi]^T$ denotes the aircraft's attitude in Euler angles, $\textbf{p} = [X,\,Y,\,Z]^T$ represents its position in the inertial reference frame, $\textbf{v}= [u_b,\,v_b,\,w_b]^T $ and $\omega = [p_b,\,q_b,\,r_b]^T$ denote the linear and angular velocities of the UAS in the body-fixed reference frame, respectively, and $u = [u_E, \, u_A,\, u_R,\, u_T]^T$ is the control input. Here, $u_E, u_A, \text{and } u_R$ are the elevator, aileron, and rudder deflections, respectively, and $u_T$ is the throttle input. The linear velocity of the UAS relative to the wind is given by $\textbf{v}_r = \textbf{v} - \textbf{v}_w$, where $\textbf{v}_w= [u_w,\,v_w,\,w_w]^T $ is the wind velocity expressed in the body frame. The constants $\texttt{I} = \mathrm{diag}(1.32,\,1.57,\,1.87)$, $\texttt{m} = 5.71~\mathrm{kg}$, and $\texttt{g} = 9.81~\mathrm{m/s^2}$ denote the moment of inertia tensor, mass of the aircraft, and acceleration due to gravity, respectively. The net aerodynamic and propulsive forces and moments acting on the aircraft in the body frame are denoted by $\bm{f}(u,\textbf{v}_r,\omega)$ and $\bm{m}(u,\textbf{v}_r,\omega)$, respectively. For details on how these forces and moments are modeled, as well as the definitions of the rotation matrix $R_{b}^I(\lambda)$, the matrix $E(\phi,\theta)$, and the values of the saturation limits, the reader is referred to \cite{muniraj2017path}.

The differential equations are expressed in state-space form with $x = [\omega^T,\, \textbf{v}^T,\, \lambda^T,\, \textbf{p}^T]^T$ as the system state, $u$ as the input, and $\textbf{v}_w$ as the process noise. For simplicity, we do not consider wind in the vertical direction, which is reasonable for low altitude flights. We are interested in controlling the UAS around a circular trajectory, where the trim airspeed is $15~\mathrm{m/s}$ and the radius of curvature is $80~\mathrm{m}$. The trim states and control inputs $(x^\ast,\,u^\ast)$ corresponding to this trajectory are computed by solving a set of nonlinear equations \cite{muniraj2017path}. The trim values for the process noise is zero. We then perform a change of variables and define $(\bar{x},\, \bar{u}) = (x - x^\ast,\, u - u^\ast)$ as the error state and control input of the UAS. The ellipsoid $\varepsilon(P)$ containing all possible values of the initial error state $\bar{x}_0$ is defined by $P = \mathrm{diag}(\pi/6,\,\pi/6,\,\pi/6,\, 2,\, 1,\, 1,\, \pi/9,\,\pi/9,\,\pi/9,\,1,\,1,\,1)^{-2}$.

To generate the error state and control input trajectories, we perform multiple closed-loop nonlinear simulations of the UAS flying along the circular trajectory over a $2~\mathrm{s}$ horizon with $\Delta t = 0.01~\mathrm{s}$. A full-state feedback LQR controller is designed for this purpose based on a linear model obtained by linearizing the nonlinear dynamics around $(x^\ast,\,u^\ast)$ using the Jacobian approach and subsequently discretizing the resulting continuous-time system using the zero-order hold method. During these simulations, the initial error state is randomly sampled from the set $\varepsilon(P)$, and the UAS is subjected to time-varying winds with components pseudorandomly generated from a uniform distribution within the range $[-2,\,2]~\mathrm{m/s}$. Measurement noise is sampled from a normal distribution with a zero mean and a standard deviation of $0.04~\mathrm{rad}$ for attitude, $0.04~\mathrm{rad/s}$ for angular rate measurements, $1~\mathrm{m/s}$ for velocity measurements, and $2~\mathrm{m}$ for position measurements. The disturbances considered in these simulations are not representative of those used for analyzing the lifted LPV model-based controllers or those encountered in real-world scenarios; rather, they are chosen to ensure that the generated data spans a sufficiently large envelope. 

\begin{table}
    \setlength{\tabcolsep}{6pt} 
    \renewcommand{\arraystretch}{1.3}
    \vspace{-0.5cm}
        \caption{Parameters and performance metrics of the lifted LPV models.} \vspace{0.25cm}
    \centering
    \begin{tabular}{c|c|c|c|c|c|c|c}
         & $N$ ($\bar{N}$)& $p$ & $\Delta t$ & $T$ & $\mathcal{L}_{dyn}$ (training) &  $\mathcal{L}_{dyn}$ (validation) & $\sqrt{ \det{Q^{-1}} }$  \\ \hline 
       Pendulum  & $20$ ($16$) & $2$ & $0.02~\mathrm{s}$ & $15$ & $7.47 \times 10^{-6}$ & $7.68 \times 10^{-6}$ & $3.50 \times 10^{-28}$    \\
       6-DOF UAS  & $25$ ($13$) & $2$ & $0.01~\mathrm{s}$  & $5$ & $1.32 \times 10^{-6}$ & $2.89 \times 10^{-6}$ & $1.08 \times 10^{-21}$    \\
    \end{tabular}
    \label{tbl_LPVmodels}
\end{table}

\begin{figure}
     \centering
     \begin{subfigure}[b]{0.48\textwidth}
         \centering
         \includegraphics[width=\textwidth]{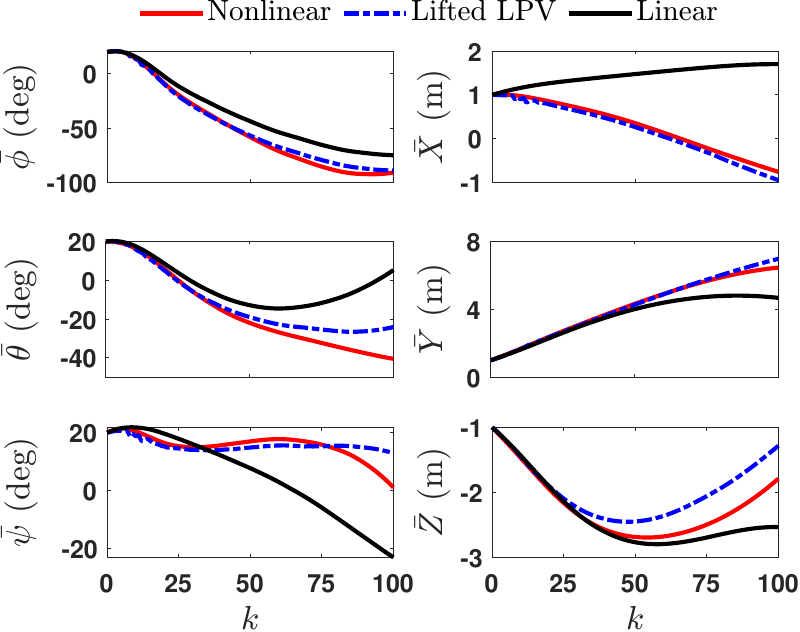}
         \caption{Vehicle Pose.}
         \label{fig_UAScomparison_pose}
     \end{subfigure}
     \hfill
     \begin{subfigure}[b]{0.48\textwidth}
         \centering
         \includegraphics[width=\textwidth]{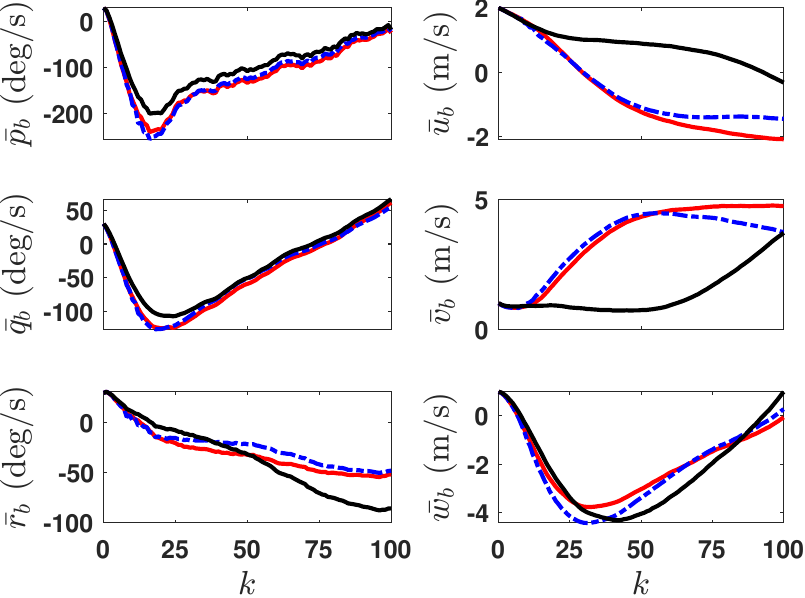}
         \caption{Velocity and Angular Rates.}
         \label{fig_UAScomparison_rates}
     \end{subfigure}
        \caption{State histories of the nonlinear, linearized, and LPV 6-DOF UAS models for the same input.}
        \label{fig_UAScomparison}
\end{figure}

Using the approach described in Section~\ref{sec_learningApproach}, we learn a lifted LPV model to approximate the nonlinear dynamics of the UAS. The details of the learned model, along with the training and validation losses associated with the model approximation, are summarized in Table~\ref{tbl_LPVmodels}. To show that the lifting-based approach generates a more accurate approximation of the nonlinear system, we compare the state trajectories of the nonlinear, linearized, and lifted LPV models under identical control and disturbance input histories and initial conditions. Multiple test cases, consisting of input histories and initial conditions, are generated through closed-loop simulations of the nonlinear model to ensure its states remain bounded. The state trajectories corresponding to a representative test case are shown in Figure~\ref{fig_UAScomparison}. Notably, the LPV model's response closely matches the behavior of the nonlinear system. Since the system is inherently unstable, the responses of the linear approximations eventually diverge from the nonlinear system due to the accumulation of errors over time. However, the lifted LPV system demonstrates a significantly delayed divergence compared to the linearized model. Hence, we expect the optimal controller designed using the lifted LPV model to yield near-optimal performance with the nonlinear system.

Using the IQC-based tuning routine discussed in Section~\ref{sec_iqctuning}, we design various controllers for the UAS. The performance output of $(G,\,\Delta)$ for analysis is defined as in \eqref{eqn_perfOut}, incorporating the error states and control inputs. The penalty weight vector is set to $\mathbf{1}$. Although this selection implies equal prioritization of all errors, a detailed examination of the scales of the UAS state and input variables reveals that position errors receive higher implicit weighting due to their larger magnitudes. For control synthesis, the penalty weight vector of the plant $(\bar{G},\,\Delta)$ is treated as an optimization variable and initialized with a value of $\mathbf{1}$. 
For synthesis, the upper bound on the $\ell_2$-norm of the disturbance $d$ in $(\bar{G},\,\Delta)$, consisting solely of process noise, is set to $50$. For closed-loop analysis of $(G,\,\Delta)$, this value is increased to $150$ to account for the inclusion of measurement noise. The standard deviation values for the original state measurements are chosen as $0.1$ for attitude and angular rate, and $0.5$ for position and velocity, with units as previously defined. The standard deviation values for the lifted state measurements, along with the bounds and rate-bounds for the scheduling parameters, are estimated using data. In the tuning routine, the perturbation $\delta_c$ for gradient computation is set to $0.05$. The line search is conducted over $20$ points with the interval boundary $\alpha_{max}$ set to $0.5$. The tuning routine is stopped when the relative decrease in the value of $\gamma$ over an iteration is less than $0.5\,\%$. To simplify the tuning process, we tune a standard LPV controller by setting $\bar{h} = 0$ in Theorem~\ref{col1}. Once the final tuned penalty weights are obtained, the synthesis is repeated with the same penalty weights but an increased horizon $\bar{h}$ to generate NSLPV controllers. The horizon is incrementally increased until no significant improvement in performance is observed. In this example, increasing $\bar{h}$ from $0$ to $2$ improved the robust performance level of the uncertain system by approximately $14\,\%$. To highlight the advantages of implementing an LPV controller, we also tune an LTI controller synthesized based on a nominal system $G_{nom}$, derived from $(\bar{G},\,\Delta)$, as discussed in Remark~\ref{rem1}. The final tuned penalty weights for both the LPV and LTI controllers are given in Table~\ref{tbl_tunedPenalties}. The ``normalized'' robust performance levels $\bar{\gamma}_{IQC}$, which we define as $\gamma/\lVert d \rVert_{\ell_2}$, obtained from IQC analysis are $2.08$, $1.72$, and $1.49$  for the LTI, LPV, and NSLPV controllers, respectively. 

\begin{table*}[!t]
\renewcommand{\arraystretch}{1.0} 
\setlength{\tabcolsep}{3.5pt}
\caption{The initial and the final (tuned) penalty weights for the lifted model-based UAS controllers.}
\centering
\begin{tabular}{|c|c|c|c|c|c|c|c|c|c|c|c|c|c|c|c|c|}
\hline
 & $c_p$ & $c_q$ & $c_r$ & $c_u$ & $c_v$ & $c_w$ & $c_\phi$ & $c_\theta$ & $c_\psi$ & $c_X$ & $c_Y$ & $c_Z$ & $c_{{u}_E}$ & $c_{{u}_A}$ & $c_{{u}_R}$ & $c_{{u}_T}$  \\
\hline
\hline
{Initial} & $1$ & $1$ & $1$ & $1$ & $1$ & $1$ & $1$ & $1$ & $1$ & $1$ & $1$ & $1$ & $1$ & $1$ & $1$ & $1$   \\
\hline
{Final (LPV)} & $0.99$ & $0.98$ & $0.99$ & $1.00$ & $1.05$ & $0.98$ & $0.98$ & $0.98$ & $0.99$ & $0.98$ & $0.99$ & $0.99$ & $1.04$ & $1.00$ & $0.98$ & $0.95$   \\
\hline
{Final (LTI)} & $1.09$ & $1.05$ & $1.02$ & $0.99$ & $0.96$ & $1.04$ & $1.04$ & $1.00$ & $1.00$ & $1.01$ & $0.89$ & $0.93$ & $1.02$ & $1.03$ & $1.01$ & $1.01$   \\
\hline
\end{tabular}
\label{tbl_tunedPenalties}
\end{table*}

\begin{figure}[t]
     \centering
     \begin{subfigure}[b]{0.49\textwidth}
         \centering
         \includegraphics[width=\textwidth]{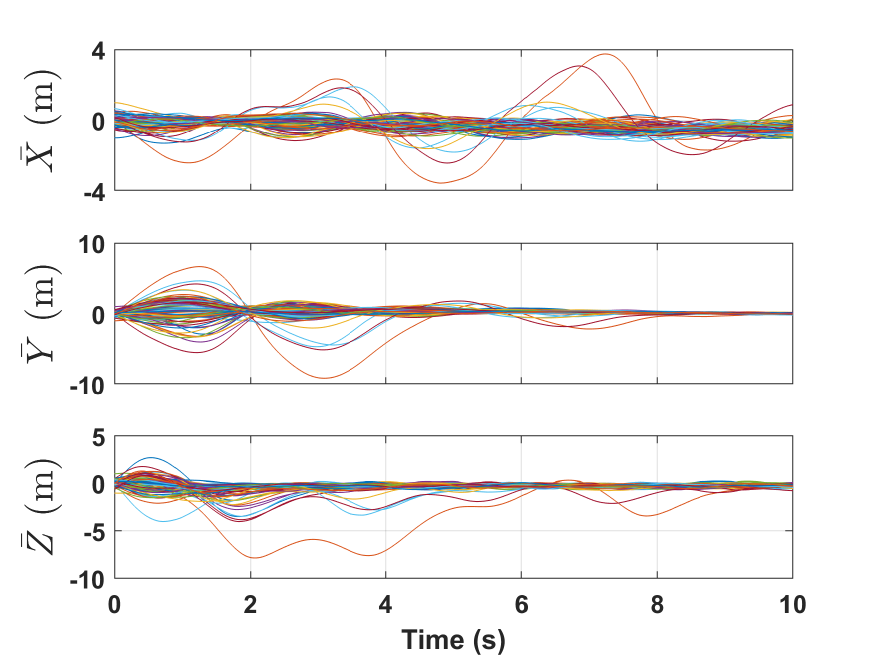}
         \caption{Linearized model-based LTI controller ($\bar{\gamma}_{sim} = 1.65$).} \vspace{0.4cm}
         \label{fig_UASCLsim_LTILinearizedCont}
     \end{subfigure}
     \hfill
     \begin{subfigure}[b]{0.49\textwidth}
         \centering
         \includegraphics[width=\textwidth]{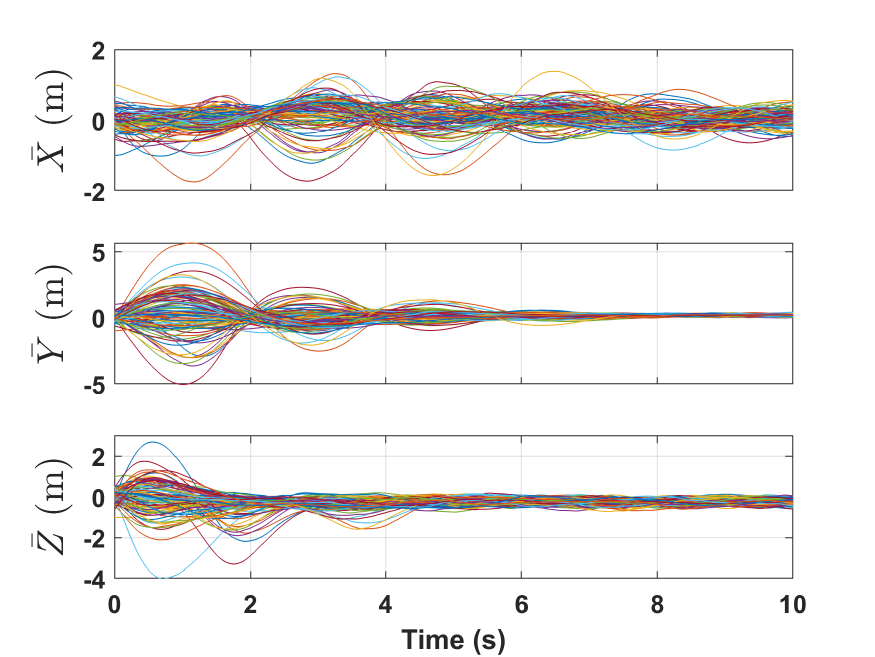}
         \caption{Lifted model-based LTI controller ($\bar{\gamma}_{sim} = 0.82$).} \vspace{0.4cm}
         \label{fig_UASCLsim_LTIKoopCont}
     \end{subfigure}
     \hfill
     \begin{subfigure}[b]{0.49\textwidth}
         \centering
         \includegraphics[width=\textwidth]{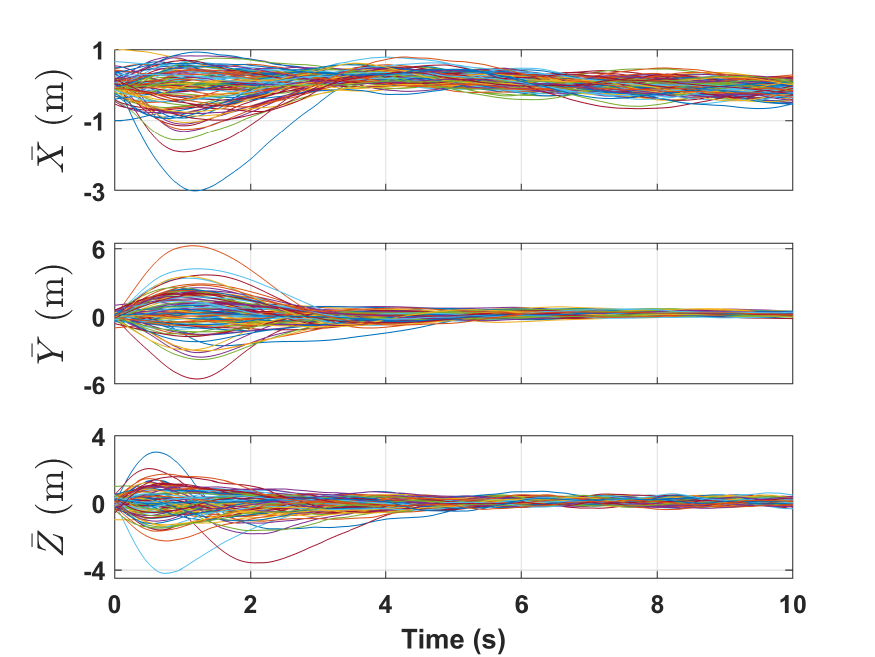}
         \caption{Lifted model-based LPV controller ($\bar{\gamma}_{sim} = 0.64$).} \vspace{0.2cm}
         \label{fig_UASCLsim_LPVKoopCont}
     \end{subfigure}
     \hfill
     \begin{subfigure}[b]{0.49\textwidth}
         \centering
         \includegraphics[width=\textwidth]{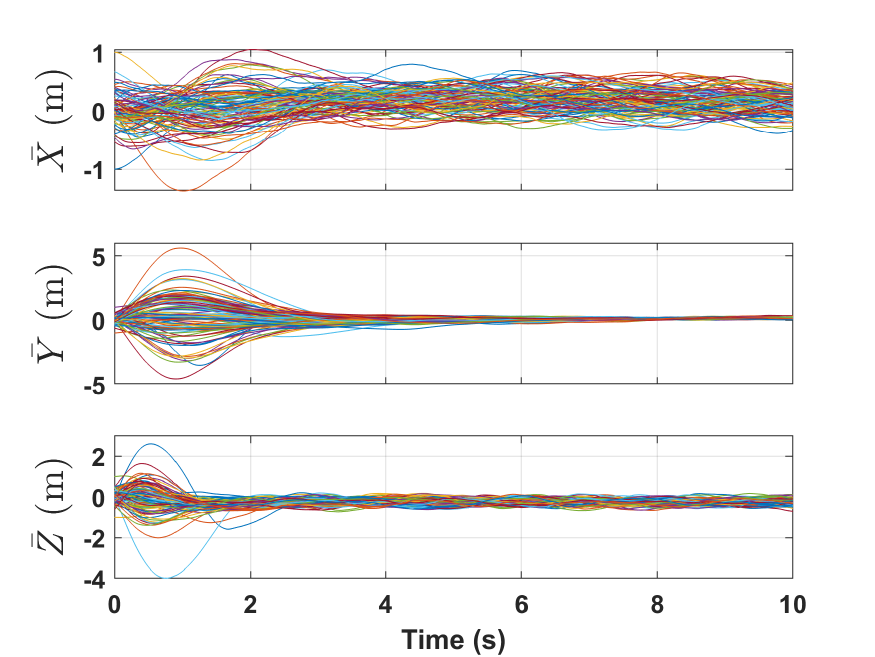}
         \caption{Lifted model-based NSLPV controller ($\bar{\gamma}_{sim} = 0.56$).} \vspace{0.2cm}
         \label{fig_UASCLsim_NSLPVKoopCont}
     \end{subfigure}
        \caption{Simulations of the 6-DOF nonlinear UAS model with the designed  full-state feedback controllers.}
        \label{fig_UASCLsim}
\end{figure}

For comparison, we also synthesize an LTI controller based on the linearized model. The initial state of this model is $x_0$, which lies in the set $\varepsilon(P)$. Hence, we modify the synthesis conditions in Corollary~\ref{col2}, as these conditions are formulated for LTI systems with initial state values lying in a set defined by two separate ellipsoids. The modified conditions are obtained by setting $\Gamma = P^{-1/2}$, $F_1 = f_{11} I_n$, and $f_{12} = 0$ in \eqref{eqn_LTISFsynthesis1}. During tuning, we conduct the analysis directly on the linearized model. In this case, the LFT interconnection $(G,\,\Delta)$ is essentially an LTI system with no scheduling parameters, i.e., $\Delta = 0$. However, the final controller obtained from this tuning fails in simulations, despite the analysis predicting a finite robust performance level. This outcome is not surprising, as the linearized model does not provide a sufficiently accurate approximation of the nonlinear system in the considered envelope, rendering the analysis inadequate for certifying the controller. To address this, we manually tune the penalty weights for the linearized plant. The manual tuning involves iteratively adjusting the penalty weights based on observations from closed-loop nonlinear simulations. The manually tuned weights are then used as the initial choice for the tuning routine, which yields a final controller with $\bar{\gamma}_{IQC}$ = 0.64.
We also attempt to tune the linearized model-based synthesis utilizing the lifted LPV model for analysis within the tuning routine. In this case, the measurement output of $(G,\,\Delta)$, as defined in \eqref{eqn_measOut}, is modified to $\hat{y}_k = \hat{x}_k + Ww_k = Cz_k +  Ww_k$, where the additive term represents the measurement noise associated with the original state. However, during gradient computation and line search within the tuning routine, the controllers synthesized based on the linearized plant frequently lead to an infeasible analysis problem or provide robust performance levels that are large and not useful, rendering the tuning ineffective. This issue arises partly from the limited accuracy of the linearized model and partly from the limitations of the tuning routine, which assumes that the closed-loop uncertain systems used for analysis are robustly stable.

We compare the closed-loop responses associated with the tuned controllers through multiple nonlinear simulations of the 6-DOF UAS, with the initial error state sampled from the set $\varepsilon(P)$. The UAS is subjected to time-varying winds, where each wind component lies within the range $[-4,\,4]$, resulting in a maximum wind speed of approximately $5.66~\mathrm{m/s}$. The measurement noise is generated using MATLAB function \texttt{randn} with the standard deviations defined previously. Identical initial conditions and disturbance sequences are used for all controllers to ensure a fair comparison. The simulations are conducted over $10~\mathrm{s}$, and we define the simulation-based robust performance level as
\begin{equation}
    \bar{\gamma}_{sim}^{~2} = \max{ \left(\sum_{k=0}^{N_k-1} e_k^Te_k \right) / \left(\sum_{k=0}^{N_k-1} d_k^Td_k \right) },
    \label{eqn_gammasim}
\end{equation}
where $N_k$ denotes the total number of time-steps in each simulation, $e_k = [\bar{x}_k^T ~~ \bar{u}_k^T]^T$ is the performance output, and the disturbance $d$ consists of wind velocity and ``normalized" unit-variance measurement noise. A few time histories of the position error, obtained from simulations, corresponding to all tuned controllers are shown in Figure~\ref{fig_UASCLsim}. The time histories of the other state variables are omitted for brevity, as they exhibit similar trends. The $\bar{\gamma}_{sim}$ value for the LTI controllers designed based on the linearized and lifted LPV models are $1.65$ and $0.82$, respectively. The robust performance levels for the LPV and the $(2,\,1)$-eventually periodic NSLPV controllers are $0.64$ and $0.56$, respectively. Notably, using the lifted LPV model for control design and analysis results in a controller with a robust performance level approximately three times smaller than that achieved by the  controller synthesized using the linearized model, underscoring the advantages of the proposed approach. We can also observe that, for the controllers designed using the lifted LPV model, the simulation-based robust performance level is smaller than the worst-case bound computed via IQC analysis. Moreover, the IQC analysis results align with nonlinear simulations when comparing the performance of the three controllers. In contrast, these observations do not hold for the LTI controller synthesized using the linearized model, which is expected given the significant discrepancies between the nonlinear and linearized dynamics.

\subsection{Double Pendulum}

The equations of motion for the double pendulum shown in Figure~\ref{fig_Pend} are given by
\begin{align*}
    \begin{split}
        \tilde{\theta} \ddot{q}_1 =& \ \theta_2(\tau + \theta_3 \dot{q}_2^2 \sin{q_2} + 2 \theta_3 \dot{q}_1\dot{q}_2\sin{q_2} 
                      - \theta_4 \texttt{g} \cos{q_1}  - \theta_5 \texttt{g} \cos{(q_1+q_2)})  \\
                     &  + (\theta_3 \dot{q}_1^2 \sin{q_2} + \theta_5 \texttt{g} \cos{(q_1+q_2)})(\theta_2 + \theta_3 \cos{q_2}),
    \end{split} \\
    \begin{split}
        \tilde{\theta} \ddot{q}_2 =&  - (\theta_2 + \theta_3 \cos{q_2})(\tau + \theta_3 \dot{q}_2^2 \sin{q_2} + 2 \theta_3 \dot{q}_1\dot{q}_2\sin{q_2}  
                            - \theta_4 \texttt{g} \cos{q_1}  - \theta_5 \texttt{g} \cos{(q_1+q_2)}) \\ & - (\theta_1 + \theta_2 + 2\theta_3 \cos{q_2}) (\theta_3 \dot{q}_1^2 \sin{q_2} + \theta_5 \texttt{g} \cos{(q_1+q_2)}),
    \end{split}
\end{align*}
where $\tau$ is the applied torque, $\tilde{\theta} = \theta_2(\theta_1 + \theta_2 + 2\theta_3 \cos{q_2})$, and $\theta_i$, for $i = 1,\dots,5$, are constants dependent on the masses, lengths, and inertia of the two links in the pendulum. The values for these constants are: $\theta_1 = 0.0308$, $\theta_2 = 0.0106$, $\theta_3 = 0.0095$, $\theta_4 = 0.2097$, and $\theta_5 = 0.0634$. The saturation limit for the torque is set to $5~\mathrm{N\text{-}m}$. The dynamics are written in state-space form with $x = (q_1,\,q_2,\,\dot{q}_1,\,\dot{q_2})$ as the state and $u = \tau$ as the input. We are interested in controlling the pendulum about an equilibrium $(x^\ast,\,u^\ast)$, where $x^\ast = (\pi/2,\,0,\,0,\,0)$ and $u^\ast = 0$. The error state and control input are defined as $(\bar{x},\, \bar{u}) = (x - x^\ast,\, u - u^\ast)$. The ellipsoid $\varepsilon(P)$ is defined by $P = \mathrm{diag}(\pi/12,\, \pi/12,\, \pi/18,\, \pi/18)^{-2}$. 

\begin{figure}[t]
\begin{minipage}[b]{0.42\textwidth}
    \centering
    \includegraphics[width=\linewidth]{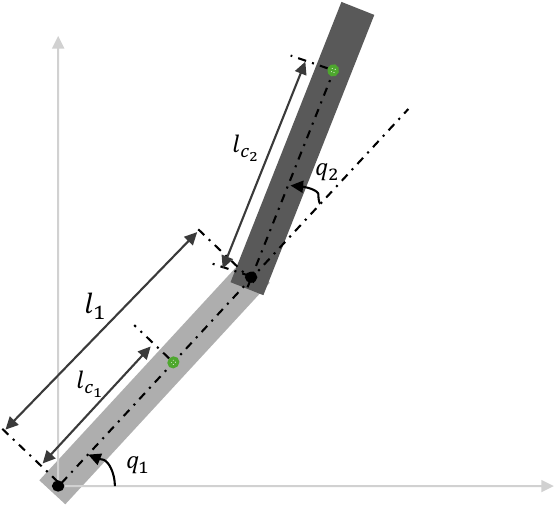}
    \captionof{figure}{Double pendulum.}
    \label{fig_Pend}
\end{minipage}
\hfill
\begin{minipage}[b]{0.5\textwidth}
    \centering
    \includegraphics[width=\linewidth]{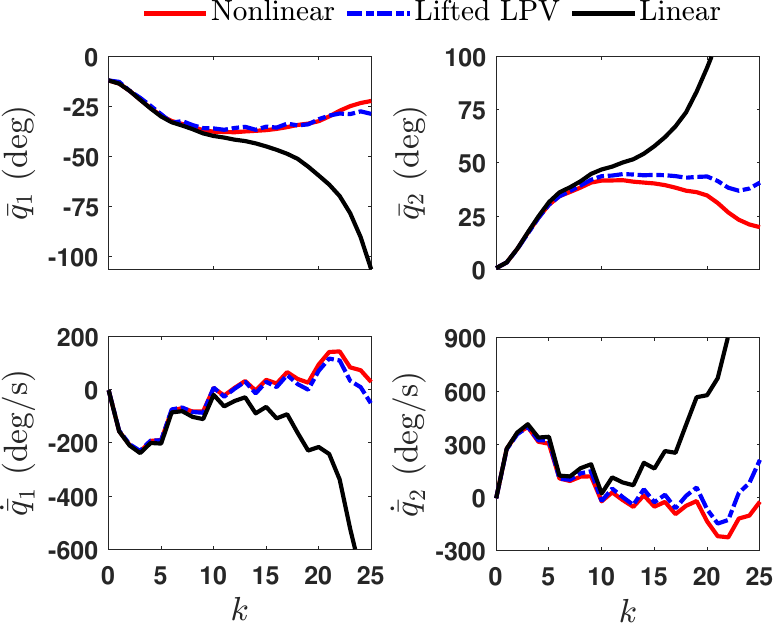}
    \captionof{figure}{State histories of the nonlinear, linearized, and LPV double pendulum models for the same input.}
    \label{fig_DPcomparison}
\end{minipage}
\end{figure}

For data generation, we perform closed-loop nonlinear simulations over a $1~\mathrm{s}$ horizon with $\Delta t = 0.02~\mathrm{s}$. An LQR controller is used to stabilize the system around the equilibrium. During simulations, the initial error state is randomly sampled from the set $\varepsilon(P)$, and state measurements are corrupted with additive white Gaussian noise, with a standard deviation of $0.02$ for each measurement noise. We also consider process noise, sampled from a uniform distribution within the range $[-0.5,\,0.5]$, which perturbs the input commands generated by the controller. During simulations, we observed that the LQR controller, designed using the linearized model, failed to stabilize the system for certain values of $\bar{x}_0$ near the boundary of $\varepsilon(P)$. These simulation results were discarded from the dataset. To include trajectories originating from such initial state values in the learning dataset, we employed a switched control approach. In this method, the reference trajectory was generated by controlling the system about a sequence of equilibrium points to drive the double pendulum from  an ``extreme'' initial state  to the target equilibrium $(x^\ast,\,u^\ast)$. For each equilibrium point, a separate LQR controller was designed to stabilize the system locally around that point. 
Since the process noise is added directly to the control input, the input $\tilde{u}_k$ of the LPV model \eqref{eqn_Koopmansys} is defined as $\bar{u}_k + v_k$ rather than $[\bar{u}_k^T, ~ v_k^T]^T$ to simplify the learning process. Table~\ref{tbl_LPVmodels} summarizes the details of the learned LPV model, while Figure~\ref{fig_DPcomparison} compares the state trajectories of the nonlinear, linearized, and lifted LPV models under identical input histories and initial conditions. In this specific case, a minimal LFT representation of the form \eqref{eqn_LFTkoop} can be derived, with $m_i = n_u = n_v$ for all $i \in \mathbb{N}_p$, $B_{1s} = \begin{bmatrix}0_{N \times n_w} & B_{0}\end{bmatrix}$, $B_{2s} = B_{0}$, $B_{1p} = \bm{1}_p \otimes \begin{bmatrix}0_{n_v \times n_w} & I_{n_v}\end{bmatrix}$, and $B_{2p} = \bm{1}_p \otimes I_{n_u}$.

For controller tuning, the penalty weight vector associated with the performance output of $(G,\,\Delta)$ is set to $\mathbf{1}$. The penalty weight vector of the plant $(\bar{G},\,\Delta)$ is initialized with a value of $\mathbf{1}$. For synthesis, the upper bound on the $\ell_2$-norm of $d$ (process noise) is set to $10$, and for analysis with both process noise and measurement noise, it is increased to $50$. The noise characteristics in the original state space are the same as those used for generating the learning dataset. The parameters in the tuning routine are identical to those used in the UAS example. Using this tuning routine, a standard LPV controller is designed for the double pendulum. The controller achieves a $\bar{\gamma}_{IQC}$ value of $3.64$. In this example, increasing the horizon $\bar{h}$ to generate NSLPV controllers did not improve the robust performance level. The final tuned penalty weight vector for the LPV controller is $[1.05,\,1.00,\,1.07,\,0.76,\,1.08]^T$. An LTI controller is also tuned based on the linearized model, starting with $c = \mathbf{1}$ as the initial choice for the penalty weight vector. The resulting tuned vector is $[1.01,\,1.00,\,0.73,\,1.07,\,0.57]^T$. However, this controller fails in some simulations. Unlike the UAS example, manual tuning did not resolve this issue, highlighting the limitations of the linearized model in this scenario. Additionally, we attempted to tune an LTI controller based on the lifted LPV model, as discussed in Remark~\ref{rem1}. However, the tuning routine failed to converge to a controller with a small robust performance level comparable to that achieved by the LPV controller. Note that  IQC analysis is conducted on the lifted LPV model, which is an approximation of the nonlinear system within the considered operating envelope. Therefore, the tuned controllers must ensure that this envelope is not violated. The large robust performance level values suggest that the controller may lead to a closed-loop response that is not contained within the operating envelope where the lifted LPV model is valid, rendering the analysis results unreliable. Simulations further reveal that this LTI controller fails, highlighting the inadequacy of the LTI control approach for the double pendulum within the considered operating envelope. 

\begin{figure}[t]
         \centering
         \includegraphics[width=0.6\textwidth]{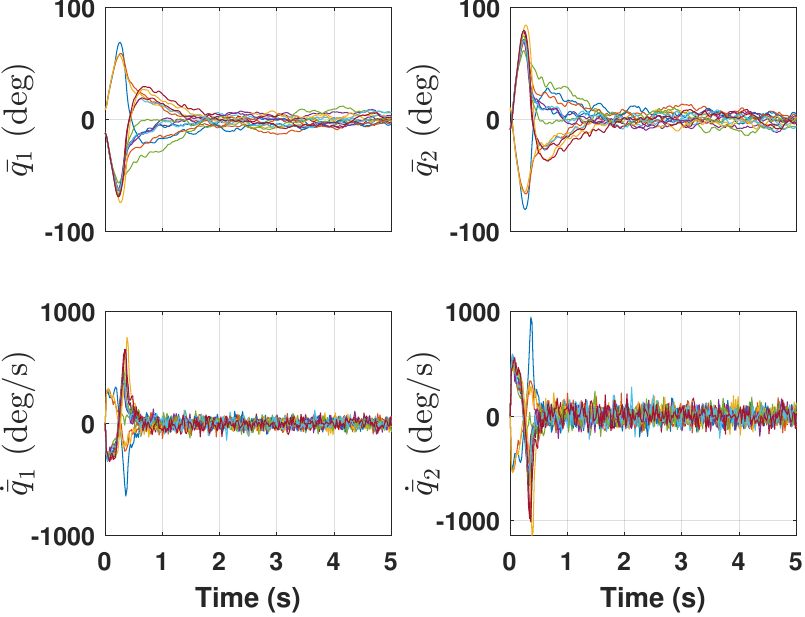}
         \caption{Simulations of the double pendulum with the full-state feedback LPV controller.}
         \label{fig_DPCLsim_LPV}
\end{figure}

\begin{figure}[t]
     \centering
     \begin{subfigure}[b]{0.49\textwidth}
         \centering
         \includegraphics[width=\textwidth]{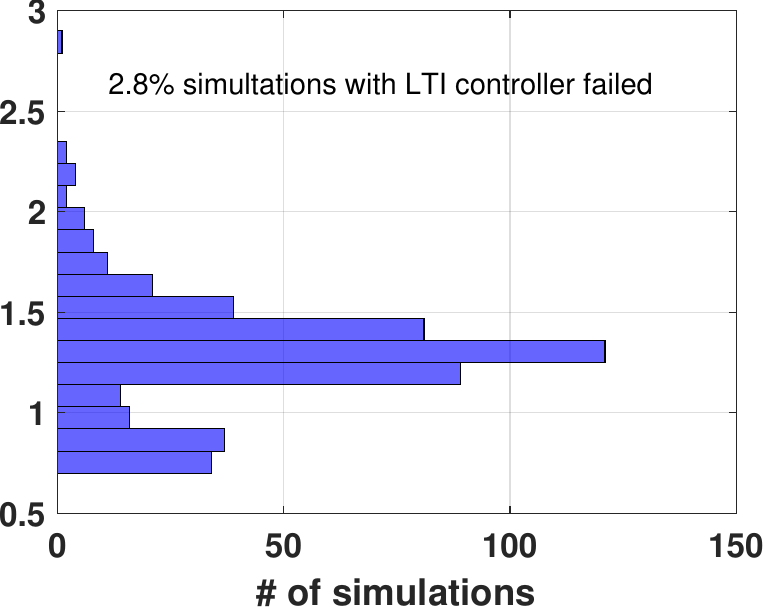}
         \caption{Linearized model-based LTI controller ($\bar{\gamma}_{sim} = 2.90$).}
         \label{fig_DPCLsim_LTIhist}
     \end{subfigure}
     \hfill
     \begin{subfigure}[b]{0.49\textwidth}
         \centering
         \includegraphics[width=\textwidth]{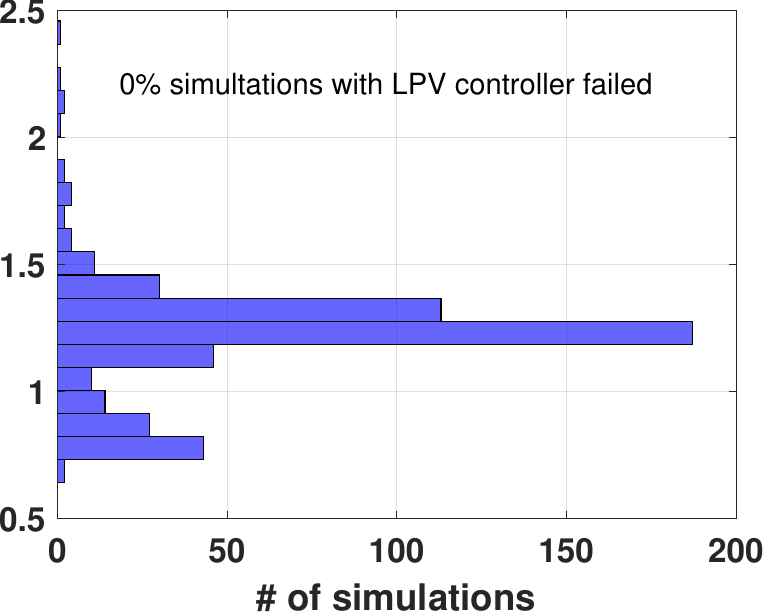}
         \caption{Lifted model-based LPV controller ($\bar{\gamma}_{sim} = 2.45$).}
         \label{fig_DPCLsim_LPVhist}
     \end{subfigure}
        \caption{Distributions of input-output norm ratios computed from simulations (excluding failed runs of LTI controller).}
        \label{fig_DPCLsim}
\end{figure}

The performance of each  controller is evaluated through closed-loop simulations of the double pendulum. A total of $500$ simulations over a $5~\mathrm{s}$ horizon are conducted for both LTI and LPV controllers under identical initial conditions and disturbances. The initial error state is sampled from $\varepsilon(P)$, and the process and measurement noise characteristics match those used during data generation. The linearized model-based LTI controller fails in $14$ out of $500$ simulations, while the lifted model-based LPV controller succeeds in all simulations. Figure~\ref{fig_DPCLsim_LPV} shows the nonlinear system response with the LPV controller for the test cases in which the LTI controller fails. For each successful simulation, the input-output norm ratio, as defined in \eqref{eqn_gammasim}, is also computed. The distribution of these ratios is shown in Figure~\ref{fig_DPCLsim}. The maximum input-output norm ratio ($\bar{\gamma}_{sim}$) comes out to be $2.90$ for the LTI controller and $2.45$ for the LPV controller, highlighting the benefits of the lifting-based LPV control approach.

\section{Limitations} \label{sec_limitations}

As with most data-driven modeling techniques, it is difficult to claim that the learned LPV model will approximate the behavior of the nonlinear system reasonably well. The model’s accuracy heavily relies on the diversity of the dataset within the envelope of interest, as well as the appropriate selection of learning hyperparameters and model architecture. If the nonlinear dynamics of the system are known, the dataset can be efficiently generated by simulating the system under a wide range of conditions. Additionally, hyperparameter search can be accelerated through parallel computation, and transfer learning techniques can further expedite the process by leveraging knowledge from related models.
One significant limitation of the lifting-based approaches is the curse of dimensionality. While simplifying the synthesis to generate static state-feedback controllers reduces the complexity of both the synthesis and controller implementation, the IQC analysis required for controller tuning can become computationally expensive as the dimension of the lifted model increases. Additionally, the complexity of the analysis grows with the number of inputs, as this directly influences the size of the perturbation operator. Finally, since the lifted LPV models are approximations, the robust performance guarantees derived for these models cannot be directly transferred to the nonlinear systems without accounting for the approximation errors.
{ Furthermore, the upper and lower bounds on the scheduling parameters, calculated from data, may constitute optimistic estimates of the true parameter bounds, thus affecting the outcomes of the robustness analysis.}

\section{Conclusion} \label{sec_conclusions}

This paper introduces a robust control design and analysis framework for nonlinear systems with uncertain initial conditions. Leveraging a lifting linearization technique, we propose a deep learning-based approach to learn LPV approximations of nonlinear systems in higher-dimensional spaces while simultaneously characterizing the uncertain initial states within the lifted state space. We further present a convex synthesis and IQC-based analysis approach for a fairly general class of NSLPV systems with uncertain initial conditions. Assuming the states are exactly measurable, the synthesis approach is simplified to generate static state-feedback controllers for the nonlinear system. The IQC analysis framework is then employed to tune these controllers, while also accounting for measurement noise. Through applications to a double pendulum and a 6-DOF UAS, we demonstrate that the lifted LPV models provide a more accurate approximation of the nonlinear systems over a larger envelope compared to linear models obtained through traditional linearization techniques. Additionally, we show that the controllers designed based on the lifted LPV model outperform those designed using linearized models. Future work will focus on addressing the errors introduced by the lifting-based LPV approximation to ensure that the robust performance guarantees derived from IQC analysis always hold for the nonlinear system within the envelope of interest.


\section*{Acknowledgments}
This material is based upon work supported by the Army Research Office under Grant Number W911NF-21-1-0250.

\bibliographystyle{unsrtnat}
\bibliography{references}


\appendix
\renewcommand{\thesection}{\Alph{section}.}
\section{Dynamic Output Feedback Synthesis} \label{app_A}

Here, we present a result on dynamic output feedback controller synthesis. Theorem~\ref{thm1} addresses a more general class of eventually periodic LFT plants with uncertain initial conditions. 
\vskip 0.1in
\begin{thm} \label{thm1}
Consider an $(h,\,q)$-eventually periodic uncertain plant $(G,\,\bm{\Delta}) = \{ (G,\,\Delta) ~|~ \Delta \in \bm{\Delta} \}$, where the interconnection $(G,\,\Delta)$  is defined as in \eqref{eqn_LFTkoop}. Suppose that the initial state $z_0$ of the plant is uncertain and can be expressed as $z_0 = \Gamma \xi$, where $\Gamma \in \mathbb{R}^{N \times s}$ and $\xi = (\xi_1, \, \dots, \, \xi_a)$, with 
$\xi_i \in \mathbb{R}^{s_i} $ and $\sum_{i=1}^{a} s_i = s \leq N$. Given some $\bar{h} \geq h$, an $(\bar{h},\,q)$-eventually periodic synthesis $(G^c,\, \bm{\Delta}^c)$ exists for this plant leading to an uncertain closed-loop system satisfying the inequality 
\begin{equation}
    \sup{\{\lVert e \rVert_{\ell_2} \mid \lVert \xi_1 \rVert_{2} \leq 1, \dots, \lVert \xi_a \rVert_{2} \leq 1, \lVert d \rVert_{\ell_2} \leq 1, \Delta \in \bm{\Delta} \}}< \gamma 
    \label{eqn_perfinequality2}
\end{equation}
 if there exist positive definite matrices $X_0(k) \in \mathbb{S}^{N}$, $X_i(k) \in \mathbb{S}^{m_i}$ for $i \in \mathbb{N}_p$, $X = R,\, S$, and $k = 0,\,1,\,\dots,\,\bar{h}+q-1$, and positive scalars $b$, $f_{11}$, $f_{12}$, $\dots$, $f_{1a}$, $f_{2}$, $g$, and $t$ such that
\begin{align}
\begin{split}
  & b + \sum_{i=1}^{a} f_{1i} + f_2 < 2\gamma, \quad \Gamma^T S_0(0) \Gamma \prec F_1,
\end{split} \label{eqn_thm1_UIC}
\\[2ex]
 &N_R(k)^T \left\{ H(k) \begin{bmatrix} R_0(k) & 0 & 0 \\ 0 & \bar{R}(k) & 0 \\ 0 & 0 & g I \end{bmatrix} H(k)^T - 
        \begin{bmatrix} R_0(k+1) & 0 & 0 \\ 0 & \bar{R}(k) & 0 \\ 0 & 0 & b I \end{bmatrix} \right\} N_R(k) \prec 0, \label{eqn_thm1_FSC}\\[2ex]
 &N_S(k)^T \left\{ H(k)^T \begin{bmatrix} S_0(k+1) & 0 & 0 \\ 0 & \bar{S}(k) & 0 \\ 0 & 0 & t I \end{bmatrix} H(k) -  
          \begin{bmatrix} S_0(k) & 0 & 0 \\ 0 & \bar{S}(k) & 0 \\ 0 & 0 & f_2 I \end{bmatrix} \right\} N_S(k) \prec 0, \label{eqn_thm1_BSC}\\[2ex]
& \begin{bmatrix} R_i(k) & I \\ I & S_i(k)\end{bmatrix} \succeq 0 , \quad  \begin{bmatrix} g & 1 \\ 1 & f_2 \end{bmatrix} \succeq 0 , \quad  \begin{bmatrix} t & 1 \\ 1 & b \end{bmatrix} \succeq 0, 
\label{eqn_thm1_coupling}
\end{align}
for $i = 0,\,1,\,\dots,\,p$ and  $k = 0,\,1,\,\dots,\,\bar{h}+q-1$, where
$R_0(\bar{h}+q) = R_0(\bar{h})$, $S_0(\bar{h}+q) = S_0(\bar{h})$,
\begin{equation*}
\begin{split}
\bar{R}(k) &= \mathrm{diag}(R_1(k),\, R_2(k),\, \dots,\, R_p(k)),  \\[2pt]
\bar{S}(k) &= \mathrm{diag}(S_1(k),\, S_2(k),\, \dots,\, S_p(k)), \\[2pt]
\mathrm{Im}\, N_R(k) &= \mathrm{Ker} \begin{bmatrix} B_{2s}(k)^T & B_{2p}(k)^T & D_{12}(k)^T \end{bmatrix},\\[2pt]
\mathrm{Im}\, N_S(k) &= \mathrm{Ker} \begin{bmatrix} C_{2s}(k) & C_{2p}(k) & D_{21}(k) \end{bmatrix}, 
\end{split}  \qquad
\begin{split}
    &H(k) = \begin{bmatrix}
    A_{ss}(k) & A_{sp}(k) & B_{1s}(k) \\ A_{ps}(k) & A_{pp}(k) & B_{1p}(k) \\ C_{1s}(k) & C_{1p}(k) & D_{11}(k) 
\end{bmatrix}, \\[2pt]
    &F_1 = \mathrm{diag}(f_{11}I_{s_1},\, f_{12}I_{s_2},\, \dots,\, f_{1a}I_{s_a}), \\[2pt]
    &N_R(k)^TN_R(k) = I, \quad N_S(k)^TN_S(k) = I.
\end{split}
\end{equation*}
\end{thm} \vskip 0.05in
\begin{proof}
This theorem is similar to Theorem~2 in our previous work \cite{farhood2021lpv}, with the key difference being that this result allows the uncertain components of the initial state to lie in separate sets by imposing distinct constraints on the components $\xi_1, \, \dots, \, \xi_a$ of the vector $\xi$. For $a = 1$, the two theorems become equivalent. The proof can be deduced from those of counterpart results in our earlier works \cite{farhood2021lpv, farhood2024robustness}. The basic idea involves constructing an $(\bar{h}+1,\,q)$-eventually periodic system $(\tilde{G},\,\tilde{\Delta})$ that is isomorphic to $(G,\, \Delta)$. This isomorphic system has a zero initial state and features $a+1$ disturbance input channels. The first $a$ channels are relevant only at $k=0$ and take values $\xi_1, \, \xi_2, \, \dots, \, \xi_a$, respectively. The $(a+1)^{th}$ channel corresponds to the exogenous disturbance $d$ and is relevant only for $k > 0$. As in \cite{farhood2008control}, the work in \cite{farhood2021lpv} adapts a key result  from \cite{pirie2002robust} to the considered problem formulation, which relates the performance inequality in  (\ref{eqn_perfinequality2}) for $a=1$ to a condition in terms of the standard $\ell_2$-induced norm performance measure. As a result, the synthesis problem can be  stated in this case as finding an $(\bar{h}+1,\,q)$-eventually periodic synthesis for the scaled uncertain system $\tilde{E}^{-1/2}(\tilde{G},\,\tilde{\bm{{\Delta}}})\tilde{F}^{-1/2}$ that renders the $\ell_2$-induced norm of the closed-loop input-output map less than one for all $\tilde{\Delta}\in\tilde{\bm{\Delta}}$. Here, $\tilde{E} = \mathrm{diag}(bI,\,bI,\,\ldots)$ and $\tilde{F} = \mathrm{diag}(f_1I,\,f_2I,\,f_2I,\,\ldots)$ are memoryless block-diagonal operators, where $b$, $f_1$, $f_2$ are positive scalars that satisfy $b+f_1+f_2 < 2\gamma$. Unlike \cite{farhood2021lpv}, the isomorphic system formulated in our problem has more than one disturbance input channels active at $k = 0$. Consequently, adapting the aforementioned result from \cite{pirie2002robust} to our case (which only gives a sufficient condition when $a>1$), the first block of $\tilde{F}$ becomes $F_1 = \mathrm{diag}(f_{11}I_{s_1},\, f_{12}I_{s_2},\, \dots,\, f_{1a}I_{s_a})$, and $f_1$ in the preceding inequality is replaced with $\sum_{i=1}^a f_{1i}$; these modifications  can also be deduced from the proof of Theorem~2 in \cite{farhood2024robustness}. The remaining conditions in the theorem statement can then be derived by directly following the steps outlined in the proof of Theorem 2 in \cite{farhood2021lpv}.
\end{proof}\vskip 0.05in

The $(\bar{h},\,q)$-eventually periodic dynamic controller $(G_c,\, {\Delta}_c)$ is defined by the following equations: 
\begin{equation}
    \begin{bmatrix} z^c_{k+1} \\ \varphi^c_k \\ u_k  \end{bmatrix} = 
    \begin{bmatrix} A^c_{ss}(k) & A^c_{sp}(k) & B^c_{s}(k) \\ A^c_{ps}(k) & A^c_{pp}(k) & B^c_{p}(k) \\ C^c_{s}(k) & C^c_{p}(k) & D^c(k) \end{bmatrix} 
    \begin{bmatrix} z^c_k \\ \vartheta^c_k \\ \hat{y}_k \end{bmatrix}, \quad \vartheta^c_k = \Delta_c(k)\varphi^c_k, \quad z^c_0 = 0,
    \label{eqn_LFTconOF}
\end{equation}
where $\Delta_c(k) = \mathrm{diag}(\delta_k^{[1]} I_{m^c_1(k)}, \, \dots, \, \delta_k^{[p]} I_{m^c_p(k)})$. If the conditions in Theorem~\ref{thm1} are feasible, this controller can be constructed from the solutions $b$, $f_2$, $R_i(k)$, $S_i(k)$ for $i = 0,\,1,\,\dots,\,p$ and $k = 0,\,1,\,\dots,\,\bar{h}+q-1$ using the procedure outlined in \cite{farhood2021lpv}.

\end{document}